\documentclass[final]{dmtcs-episciences}

\usepackage[utf8]{inputenc}
\usepackage{subfigure}
\usepackage[english]{babel}
\usepackage{array}
\usepackage{amsthm}
\usepackage{amsmath,amssymb,amsfonts,mathtools}
\usepackage{framed}
\usepackage{xcolor}
\usepackage{multirow}
\usepackage{accents}
\usepackage{xspace}
\usepackage{multirow,array}
\usepackage{graphicx}
\usepackage{colortbl}
\usepackage{microtype}
\usepackage{hyperref}
\usepackage{booktabs,caption}
\usepackage[flushleft]{threeparttable}
\usepackage[T1]{fontenc}
\usepackage{tikz}
\usetikzlibrary{arrows}

\tikzset{
  treenode/.style = {align=center, inner sep=0pt, text centered,
    font=\sffamily},
  arn_n/.style = {treenode, circle, white, font=\sffamily\bfseries, draw=black,
    fill=black, text width=1.5em},
  arn_r/.style = {treenode, circle, black ,font=\sffamily\bfseries, draw=black,  fill=white, text width=1.5em, very thick},
  arn_x/.style = {treenode, rectangle ,font=\sffamily\bfseries, draw=red,
    minimum width=1.5em, minimum height=1.5em}
}
\hypersetup{
    pdftitle={}, 
    pdfauthor={}, 
    colorlinks=true, 
    linkcolor=black, 
    citecolor=blue, 
    filecolor=blue, 
    urlcolor=black
}

\hypersetup{colorlinks,linkcolor={black},citecolor={black},urlcolor={black}}

\newtheorem{theorem}{Theorem}[section]

\newtheorem{claim}[theorem]{Claim}
\newtheorem{defi}[theorem]{Definition}

\newtheorem{observation}[theorem]{Observation}
\newtheorem{note}[theorem]{Note}
\newtheorem{notation}[theorem]{Notation}

\newtheorem{open}[theorem]{Open Problem}

\definecolor{lgreen}{rgb}{0.6, 1.0, 0.0}

\newcommand{\N}{\mathbb{N}}

\newcommand{\zone}{\{0, 1\}}

\newcommand{\sym}{\mathsf{S}}

\renewcommand{\deg}{\mathrm{deg}}
\newcommand{\adeg}{\widetilde{\mathrm{deg}}}

\newcommand{\dqc}{\mathsf{D}}
\newcommand{\UC}{\mathsf{UC}}
\newcommand{\rqc}{\mathsf{R}}
\newcommand{\qqc}{\mathsf{Q}}
\newcommand{\roc}{\mathsf{R_{0}}}
\newcommand{\rac}{\mathsf{R}}
\newcommand{\cert}{\mathsf{C}}
\newcommand{\qec}{\mathsf{Q_{E}}}

\newcommand{\flip}{\mathsf{flip}}

\newcommand{\swap}{\mathsf{Swap_{\frac{1}{2}}}}

\newcommand{\AND}{\mathsf{AND}}
\newcommand{\NAND}{\mathsf{NAND}}
\newcommand{\OR}{\mathsf{OR}}
\newcommand{\MAJORITY}{\mathsf{MAJORITY}}

\newcommand{\PARITY}{\mathsf{PARITY}}

\newcommand{\bp}{\text{BPointer}}
\newcommand{\lp}{\text{LPointer}}
\newcommand{\rp}{\text{RPointer}}

\newcommand{\vl}{\text{Value}}
\newcommand{\cl}{\text{cell}}

\newcommand{\NW}{\mathsf{NW}}
\newcommand{\ku}{\mathsf{K}}
\newcommand{\RUB}{\mathsf{RUB}}
\newcommand{\GSS}{\mathsf{GSS}}
\newcommand{\desentransform}{\mathsf{DT}}
\newcommand{\abbl}{\mathsf{A1}}
\newcommand{\mabbl}{\mathsf{ModA1}}
\newcommand{\De}{\mathsf{Dec}}

\newcommand{\ksum}{\mathsf{k\mbox{-}sum}}

\newcommand{\encodedksum}{\mathsf{ENC\mbox{-}k\mbox{-}Sum}}
\newcommand{\blockksum}{\mathsf{Block\mbox{-}k\mbox{-}Sum}}
\newcommand{\type}{\mathsf{Type}}
\newcommand{\bt}{\mathsf{Bt}}

\newcommand{\Gr}{\mathsf{G}}

\newcommand{\ucmin}{\mathsf{UC_{min}}}

\renewcommand{\hat}{\widehat}

\renewcommand{\Tilde}{\widetilde}

\usepackage[round]{natbib}
\author[Chakraborty, Kayal, Paraashar]{Sourav Chakraborty\affiliationmark{1}\thanks{https://orcid.org/0000-0001-9518-6204}
  \and Chandrima Kayal\affiliationmark{2}\thanks{https://orcid.org/0009-0006-4827-3640}
  \and Manaswi Paraashar\affiliationmark{3}\thanks{https://orcid.org/0009-0005-3805-5095}}
\title[Transitive Functions]{Separations between Combinatorial Measures for Transitive Functions\footnote{A preliminary version of this work appeared in ICALP, 2022~\cite{SKP22}.}}
\affiliation{ 
  Indian Statistical Institute, Kolkata, India.\\
 Université Paris Cité, CNRS, IRIF, Paris, France.\\
  University of Copenhagen, Copenhagen, Denmark. }
\begin{document}

\publicationdata{vol. 27:2}{2025}{18}{10.46298/dmtcs.11133}{2023-03-30; 2023-03-30; 2024-10-13; 2025-05-22}{2025-06-07}


\maketitle
\begin{abstract}
The role of symmetry in Boolean functions $f:\zone^n \to \zone$ has been extensively studied in complexity theory. For example, symmetric functions, that is, functions that are invariant under the action of $\sym_n$, is an important class of functions in the study of Boolean functions. A function $f:\zone^n \to \zone$ is called transitive (or weakly-symmetric) if there exists a transitive sub-group $\Gr$ of $\sym_n$ such that $f$ is invariant under the action of $\Gr$. In other words, the value of a transitive function remains unchanged even after the input bits of $f$ are moved around according to some permutation $\sigma\in \Gr$. Understanding various complexity measures of transitive functions has been a rich area of research for the past few decades.


This work investigates relations and separations between various complexity measures for the class of transitive functions. A class of functions called ``pointer functions'' is well known for demonstrating several state-of-the-art separations for general Boolean functions. The main contribution of this work is to extend this technique to transitive functions, constructing new functions that preserve the structural properties of pointer functions while being transitive. In particular, this allows us to show separations between query complexity and other measures for the class of transitive functions.

Our results advance the understanding of the transitivity of Boolean functions and highlight the utility of certain transitive groups, which may be of independent interest in mathematics and theoretical computer science.

A comprehensive summary of the relationships between combinatorial measures for transitive functions is presented, modeled after the known table for general Boolean functions.

\end{abstract}

\section{Introduction}

 For a Boolean function $f: \zone^n \to \zone$ what is the relationship between its various combinatorial measures, like deterministic query complexity ($\dqc(f)$),  bounded-error randomized and quantum query complexity ($\rqc(f)$ and $\qqc(f)$ respectively),  zero -randomized query complexity ($\roc(f)$), exact quantum query complexity ($\qec(f)$),  sensitivity ($\mathsf{s}(f)$), block sensitivity ($\mathsf{bs}(f)$), certificate complexity ($\cert(f)$), randomized certificate complexity ($\mathsf{RC}(f)$), unambiguous certificate complexity ($\mathsf{UC}(f)$),  degree ($\deg(f)$), approximate degree ($\adeg(f))$ and spectral sensitivity ($\lambda(f)$)\footnote{We provide the formal definitions of the various measures used in this paper in Section~\ref{sec: complexity measures of Boolean functions}}?
 For over three decades, understanding the relationships between these measures has been an active area of research in computational complexity theory.  These combinatorial measures have applications in many other areas of theoretical computer science, and thus the above question takes a central position. 

In the last couple of years, some of the more celebrated conjectures have been answered - like the quadratic relation between sensitivity and degree of Boolean functions \cite{Huang}.  We refer the reader to the survey~\cite{BW} for an introduction to this area. 

Understanding the relationship between various combinatorial measures involves two parts:
\begin{itemize}
    \item Relationships - proving that one measure is upper-bounded by a function of another measure. For example, for any Boolean function $f$, $\deg(f)\leq \mathsf{s}(f)^2$ and $\dqc(f) \leq \rqc(f)^3$.
 
    \item Separations - constructing functions that demonstrate separation between two measures. For example, there exists a Boolean function $f$ with $\deg(f) \geq \mathsf{s}(f)^2$. Also, there exists another Boolean function $g$ with $\dqc(g) \geq \rqc(g)^2$.

\end{itemize} 
Obtaining tight bounds between pairs of combinatorial measures - that is when the relationship and the separation results match - is the holy grail of this area of research. The current best-known results for different pairs of functions have been nicely compiled in \cite{ABK+}. 

For special classes of Boolean functions, the relationships and the separations might be different than those of general Boolean functions. For example, while it is known that there exists $f$ such that $\mathsf{bs}(f) = \Theta(\mathsf{s}(f)^2)$ \cite{Rub95}, for a 
symmetric function a tighter result is known, $\mathsf{bs}(f) = \Theta(\mathsf{s}(f))$. The best-known relationship of $\mathsf{bs}(f)$ for general Boolean functions is $\mathsf{s}(f)^4$ \cite{Huang}.
How the various measures behave for different classes of functions has been studied since the dawn of this area of research.

\vspace{0.5em}

\noindent\textbf{Transitive Functions:} One of the most well-studied classes of Boolean functions is that of the transitive functions. A function $f:\zone^n \to \zone$ is transitive if there is a transitive group $G\leq \sym_n$ such that the function value remains unchanged even after the indices of the input is acted upon by a permutation from $G$.  Note that, when $G = \sym_n$ then the function is symmetric. Transitive functions (also called ``weakly symmetric'' functions) have been studied extensively in the context of various complexity measures. This is because symmetry is a natural measure of the complexity of a Boolean function. 
It is expected that functions with more symmetry must have less variation among the various combinatorial measures. A recent work \cite{DCK+20} has studied the functions under various types of symmetry in terms of quantum speedup. So, studying functions in terms of symmetry is important in various aspects.

For example, for symmetric functions, where the transitive group is $\sym_n$, most of the combinatorial measures become the same up to a constant 
\footnote{There are still open problems on the tightness of the constants.}. Another example of transitive functions is the graph properties. The input is the adjacency matrix, and the transitive group is the graph isomorphism group acting on the bits of the adjacency matrix.  \cite{Turan84, Sun11, LiS17, GaoMSZ13}  tried to obtain tight bounds on the relationship between sensitivity and block sensitivity for graph properties. They also tried to answer the following question: how low can sensitivity and block sensitivity go for graph properties?

In~\cite{SYZ04, Chakraborty11, Sun07, Drucker11} it has been studied how low can the combinatorial measures go for transitive functions. The behavior of transitive functions can be very different from general Boolean functions. For example, it is known that there are Boolean functions for which the sensitivity is as low as $\Theta(\log n)$ where $n$ is the number of effective variables\footnote{A variable is effective if the function is dependent on it.}, it is known (from  \cite{Sun07} and \cite{Huang}) that if $f$ is a transitive function on $n$ effective variables then its sensitivity $s(f)$ is at least $\Omega(n^{1/12})$\footnote{It is conjectured that the sensitivity of a transitive function is $\Omega(n^{1/3})$.}. Similar behavior can be observed in other measures too. For example, it is easy to see that for a transitive function, the certificate complexity is $\Omega(\sqrt{n})$, while the certificate complexity for a general Boolean function can be as low as $O(\log n)$. In Table~\ref{tab:my_label 3} we have summarized the best-known separations of the combinatorial measures for transitive functions. 

A naturally related question is: 
\begin{center}
    {\em What are tight relationships between various pairs of combinatorial measures for transitive functions?}
\end{center}

By definition, the known relationship results for general functions hold for transitive functions, but tighter relationships may be obtained for transitive functions. On the other hand, the existing separations don't extend easily since the example used to demonstrate separation between certain pairs of measures may not be transitive. Some of the most celebrated examples are not transitive. For example, some of the celebrated function constructions like the pointer function in \cite{ABBL+17}, used for demonstrating tight separations between various pairs like $\dqc(f)$ and $\roc(f)$, are not transitive. Similarly, the functions constructed using the cheat sheet techniques \cite{ABK16} used for separation between quantum query complexity and degree, or approximate degree, are not transitive. 
Constructing transitive functions that demonstrate tight separations between various pairs of combinatorial measures is very challenging. 

\vspace{0.5em} 

\noindent\textbf{Our Results: } We try to answer the above question for various pairs of measures. More precisely, our main contribution is to construct transitive functions that have similar complexity measures as the {\em pointer functions}. Hence for those pairs of measures where pointer functions can demonstrate separation for general functions, we prove that transitive functions can also demonstrate similar separation. 

Our results and the current known relations between various pairs of complexity measures of transitive functions are compiled in Table~\ref{table: main_table}. This table is along the lines of the table in~\cite{ABK+} where the best-known relations between various complexity measures of general Boolean functions were presented.

Deterministic query complexity and zero-error randomized query complexity are two of the most basic measures and yet the tight relation between these measures was not known until recently. In \cite{Snir85} they showed that for the ``balanced $\NAND$-tree'' function, $\widetilde\wedge$-tree, $\dqc(\widetilde\wedge\mbox{-tree}) \geq \roc(\widetilde\wedge\mbox{-tree})^{1.33}$. Although the function  $\widetilde\wedge$-tree is transitive, the best-known relationship was quadratic, that is for all Boolean function $f$, $\dqc(f) = O(\roc(f)^{2})$. In \cite{ABBL+17} a new function, $\abbl$, was constructed for which deterministic query complexity and zero-error randomized query complexity can have a quadratic separation between them, and this matched the known relationship results. The function in \cite{ABBL+17} was a variant of the pointer functions - a class of functions introduced by~\cite{GPW18}  that has found extensive usage in showing separations between various complexity measures of Boolean functions. Note that this classes of functions are defined on non-Boolean domain, we can make such function Boolean using the canonical Boolean encoding which is of size at most log of the size of the alphabets we are using for non-Booelan domain. The function, $\abbl$, also gave (the current best-known) separations between deterministic query complexity and other measures like quantum query complexity and degree, but the function $\abbl$ is not transitive. 
Using the $\abbl$ function
we construct a transitive function that demonstrates a similar gap between deterministic query complexity and zero-error randomized query complexity, quantum query complexity, and degree.

\begin{theorem}
\label{thm:R_0vsD}
There exists a transitive function $F_{\ref{thm:R_0vsD}}$ such that
\begin{align*}
    \dqc(F_{\ref{thm:R_0vsD}}) = \Tilde{\Omega}(\qqc(F_{\ref{thm:R_0vsD}})^4), \hspace{2em}
    \dqc(F_{\ref{thm:R_0vsD}}) = \Tilde{\Omega}(\roc(F_{\ref{thm:R_0vsD}})^{2}), \hspace{2em}
    \dqc(F_{\ref{thm:R_0vsD}}) = \Tilde{\Omega}(\deg(F_{\ref{thm:R_0vsD}})^{2}).
\end{align*}
\end{theorem}

The proof of Theorem~\ref{thm:R_0vsD} is presented in Section~\ref{sec: variant 1}. In \cite{ABBL+17, DHT17} various variants of the pointer function have been used to show separation between other pairs of measures like $\roc$ with $\rac$, $\qec$, $\deg$, and $\qqc$,  $\rac$ with $\adeg$, $\deg$, $\qec$ and sensitivity. Using similar techniques as Theorem~\ref{thm:R_0vsD} one can construct transitive versions of other variants of pointer functions(from \cite{ABBL+17, DHT17}) which give matching separations to the best-known separations of general functions. The construction of these functions, though more complicated and involved, are similar in flavor to that of $F_{\ref{thm:R_0vsD}}$. 



Using standard techniques, we can also obtain the following theorems as corollaries to Theorem~\ref{thm:R_0vsD}.



\begin{theorem}
\label{theo: intro sensitivity separations}
There exists transitive functions $F_{\ref{theo: intro sensitivity separations}}$ such that $\dqc(F_{\ref{theo: intro sensitivity separations}}) = \Tilde{\Omega}(\mathsf{s}(F_{\ref{theo: intro sensitivity separations}})^3)$.
\end{theorem}

Our proof techniques also help us make transitive versions of other functions like that used in~\cite{ABK16} to demonstrate the gap between $\qqc$ and certificate complexity. 

\begin{theorem}
\label{thm:Qvs C}
There exists a transitive function $F_{\ref{thm:Qvs C}}$ such that
$\qqc(F_{\ref{thm:Qvs C}}) = \widetilde{\Omega}(\cert(F_{\ref{thm:Qvs C}})^{2})$.

\end{theorem}

All our results are compiled (and marked in green) in Table~\ref{table: main_table}.

One would naturally ask what stops us from constructing transitive functions analogous to the other functions, like cheat sheet-based functions. In fact, one could ask why to use ad-hoc techniques to construct transitive functions (as we have done in most of our proofs) and instead, why not design a unifying technique for converting any function into a transitive function that would display similar properties in terms of combinatorial measures \footnote{In~\cite{DCK+20} they have demonstrated a technique that can be used for constructing a transitive partial function that demonstrates gaps (between certain combinatorial measures) similar to a given partial function that need not be transitive, but their construction need not construct a total function even when the given function is total.}. 
If one could do so, all the separation results for general functions (in terms of separation between pairs of measures) would translate to separation for transitive functions. We leave those directions for future work.

\begin{table}[t]
\scalebox{0.66}{
  \begin{threeparttable}
    \caption{\large best-known separations between combinatorial measures for transitive functions.}\label{table: main_table}
\centering
\begin{tabular}
{  m{0.45cm} ||c | c| c | c | c | c | c | c | c | c | c | c | c | c | c | c |} 
& & & & & & & & & & & &\\
& $\dqc$ & $\roc$ & $\rqc$ & $\mathsf{C}$ & $\mathsf{RC}$ & $\mathsf{bs}$ & $\mathsf{s}$ & $\lambda$ & $\qec$ & $\mathsf{deg}$ & $\qqc$ & $\mathsf{\widetilde{deg}}$ \\ 
\hline
\hline
{$\dqc$} & \cellcolor{black} &\cellcolor{green}2  ;  2 &\cellcolor{lgreen} 2  ;  3 & 2  ;  2 &2  ;  3&  2  ;  3& \cellcolor{lgreen}3  ;  6 & \cellcolor{lightgray}4  ;  6 & \cellcolor{lightgray}2  ;  3 & \cellcolor{lgreen}2  ;  3 &\cellcolor{green} 4  ;  4 & \cellcolor{lightgray} 4  ;  4 \\
&\cellcolor{black}  &  \cellcolor{green}T:\ref{thm:R_0vsD}& \cellcolor{lgreen}T:\ref{thm:R_0vsD}& $\land\circ\lor$ & $\land\circ\lor$ &$\land\circ\lor$ & \cellcolor{lgreen}T:\ref{theo: intro sensitivity separations} &\cellcolor{lightgray}  &\cellcolor{lightgray}&\cellcolor{lgreen}T:\ref{thm:R_0vsD}  &\cellcolor{green}T:\ref{thm:R_0vsD}  & \cellcolor{lightgray} \\

\hline
\multirow{1}{0pt}{$\roc$} & 1,1 & \cellcolor{black}& \cellcolor{lightgray}2  ;  2 & 2  ;  2& 2  ;  3  & 2  ;  3 & \cellcolor{lightgray}3  ;  6& \cellcolor{lightgray}4  ;  6 &\cellcolor{lightgray} 2  ;  3 &\cellcolor{lightgray} 2  ;  3 & \cellcolor{lightgray}3  ;  4 &\cellcolor{lightgray} 4 ; 4 \\
& $\oplus$ & \cellcolor{black} & \cellcolor{lightgray} & $\land\circ\lor$ & $\land\circ\lor$ & $\land\circ\lor$ & \cellcolor{lightgray}
 & \cellcolor{lightgray}& \cellcolor{lightgray} & \cellcolor{lightgray} & \cellcolor{lightgray} & \cellcolor{lightgray} \\

\hline
\multirow{2}{0pt}{$\rqc$} & 1 ; 1 & 1 ; 1 & \cellcolor{black}& 2 ; 2 & 2 ; 3 & 2 ; 3 &\cellcolor{lightgray} 3 ; 6 & \cellcolor{lightgray} 4 ; 6 & \cellcolor{lightgray}1.5 ; 3 & \cellcolor{lightgray}2 ; 3 &\cellcolor{yellow}2 ; 4 & \cellcolor{lightgray}4 ; 4\\
&$\oplus$  & $\oplus$ & \cellcolor{black} & $\land\circ\lor$&  $\land\circ\lor$& $\land\circ\lor$ & \cellcolor{lightgray} &\cellcolor{lightgray} &\cellcolor{lightgray} &\cellcolor{lightgray}  & \cellcolor{yellow}$\wedge$ &\cellcolor{lightgray}  \\
\hline

\multirow{2}{0pt}{$\mathsf{C}$} & 1 ; 1 & 1 ; 1 & 1 ; 2 & \cellcolor{black} & 2 ; 2 & 2 ; 2 &\cellcolor{yellow} 2 ; 5 & \cellcolor{yellow}2 ; 6 & 1.15 ; 3 & 1.63 ; 3 & 2 ; 4& 2 ; 4 \\
&$\oplus$  & $\oplus$ & $\oplus$  & \cellcolor{black} & {\scriptsize \cite{GSS16}} & {\scriptsize \cite{GSS16}}  & \cellcolor{yellow} {\scriptsize \cite{Rub95}} & \cellcolor{yellow} $\wedge$ & {\scriptsize \cite{Ambainis16}} &  {\scriptsize \cite{NW94}} & $\wedge$ & $\wedge$ \\

\hline
\multirow{2}{0pt}{$\mathsf{RC}$} & 1 ; 1 & 1 ; 1 & 1 ; 1 & 1 ; 1 & \cellcolor{black} & 1.5 ; 2 & 2 ; 4 &  2 ; 4 &1.15 ; 2 & 1.63 ; 2 & 2 ; 2 & 2 ; 2\\
&$\oplus$ & $\oplus$  & $\oplus$  & $\oplus$  &\cellcolor{black}  &{\scriptsize \cite{GSS16}} & {\scriptsize \cite{Rub95}} & $\land$& {\scriptsize \cite{Ambainis16}} & {\scriptsize \cite{NW94}} & $\wedge$ & $\wedge$ \\
\hline
\multirow{2}{0pt}{$\mathsf{bs}$} & 1 ; 1 & 1 ; 1 & 1 ; 1 & 1 ; 1 & 1 ; 1 & \cellcolor{black} & 2 ; 4 &  2 ; 4 & 1.15 ; 2& 1.63 ; 2 & 2, 2 & 2 ; 2\\
&$\oplus$  & $\oplus$ & $\oplus$ &$\oplus$  &  $\oplus$&  \cellcolor{black} & {\scriptsize \cite{Rub95}} & $\land$ & {\scriptsize \cite{Ambainis16}} & {\scriptsize\cite{NW94}} & $\wedge$ & $\wedge$ \\

\hline
\multirow{2}{0pt}{$\mathsf{s}$} & 1 ; 1 & 1 ; 1 & 1 ; 1 & 1 ; 1 & 1 ; 1 & 1 ; 1 & \cellcolor{black} & 2 ; 2 &1.15 ; 2  & 1.63 ; 2 & 2, 2 & 2 ; 2 \\
&$\oplus$  & $\oplus$ & $\oplus$ & $\oplus$ & $\oplus$ & $\oplus$ & \cellcolor{black} &$\land$  &{\scriptsize \cite{Ambainis16}} &  {\scriptsize \cite{NW94}}  & $\wedge$ & $\wedge$ \\
\hline
\multirow{2}{0pt}{$\lambda$} & 1 ; 1 & 1 ; 1 & 1 ; 1 & 1 ; 1 & 1 ; 1  & 1 ; 1 & 1 ; 1 & \cellcolor{black} &  1 ; 1 & 1 ; 1 & 1 ; 1 & 1 ; 1\\
& $\oplus$ & $\oplus$ &  $\oplus$& $\oplus$ &  $\oplus$&$\oplus$  &$\oplus$  & \cellcolor{black} & $\oplus$  & $\oplus$ & $\oplus$ &  $\oplus$\\
\hline

\multirow{2}{0pt}{$\qec$} & 1 ; 1 & 1.33 ; 2 & 1.33 ; 3 & 2 ; 2 & 2 ; 3 & 2 ; 3 & \cellcolor{yellow}2 ; 6 &\cellcolor{yellow} 2 ; 6 & \cellcolor{black} &\cellcolor{yellow} 1 ; 3 & 2 ; 4 &\cellcolor{yellow} 1 ; 4 \\
& $\oplus$ & $\widetilde{\wedge}$-tree & $\widetilde{\wedge}$-tree & $\land\circ\lor$ & $\land\circ\lor$ &  $\land\circ\lor$ & \cellcolor{yellow}T:\ref{thm:Qvs C} & \cellcolor{yellow}T:\ref{thm:Qvs C} & \cellcolor{black} & \cellcolor{yellow} $\oplus$ & $\wedge$& \cellcolor{yellow} $\oplus$\\
\hline

\multirow{2}{0pt}{$\mathsf{deg}$} & 1 ; 1 & 1.33 ; 2 & 1.33 ; 2 & 2 ; 2 & 2 ; 2 &2 ; 2 &2 ; 2 & 2 ; 2 & 1 ; 1 & \cellcolor{black} & 2 ;  2&2 ; 2\\
& $\oplus$ & $\widetilde{\wedge}$-tree & $\widetilde{\wedge}$-tree&  $\land\circ\lor$&  $\land\circ\lor$& $\land\circ\lor$ & $\land\circ\lor$ &  $\wedge$ & $\oplus$& \cellcolor{black} & $\wedge$ &   $\wedge$\\

\hline
\multirow{2}{0pt}{$\qqc$} & 1 ; 1 & 1 ; 1 & 1 ; 1 &\cellcolor{green}2  ; 2 &\cellcolor{lgreen}2  ; 3 &\cellcolor{lgreen}2  ; 3 &\cellcolor{yellow} 2 ; 6 & \cellcolor{yellow}2 ; 6 & 1, 1 & \cellcolor{yellow}1 ; 3 &  \cellcolor{black} &\cellcolor{yellow} 1 ; 4\\
& $\oplus$ & $\oplus$ & $\oplus$ &\cellcolor{green}T:\ref{thm:Qvs C}  &\cellcolor{lgreen}T:\ref{thm:Qvs C}  & \cellcolor{lgreen}T:\ref{thm:Qvs C}  & \cellcolor{yellow}T:\ref{thm:Qvs C} & \cellcolor{yellow}T:\ref{thm:Qvs C}& $\oplus$ & \cellcolor{yellow}$\oplus$  &  \cellcolor{black} & \cellcolor{yellow} $\oplus$\\
\hline

\multirow{2}{0pt}{$\mathsf{\widetilde{deg}}$} & 1 ; 1 & 1 ; 1 & 1 ; 1 & \cellcolor{yellow}  1 ; 2  & \cellcolor{yellow} 1 ; 2 & \cellcolor{yellow}  1 ; 2 & \cellcolor{yellow}  1 ; 2 & \cellcolor{yellow} 1 ; 2  & 1 ; 1 & 1 ; 1 & 1 ; 1&  \cellcolor{black} \\
& $\oplus$  & $\oplus$ &  $\oplus$& \cellcolor{yellow} $\oplus$ & \cellcolor{yellow} $\oplus$  & \cellcolor{yellow} $\oplus$  & \cellcolor{yellow} $\oplus$  & \cellcolor{yellow} $\oplus$  & $\oplus$ & $\oplus$ & $\oplus$ &  \cellcolor{black} \\
\hline

\end{tabular}
\begin{tablenotes}[flushleft]
\large 
\item[(1)] Entry $a ; b$ in row $A$ and column $B$ represents: 
for any \emph{transitive} function $f$,
$A(f) = O(B(f))^{b + o(1)}$, and there exists a \emph{transitive} function $g$ such that $A(g) = \Omega(B(g))^a$. 
\item[(2)] Cells with a green background are those for which we constructed new \emph{transitive functions} to demonstrate separations that match the best-known separations for general functions. The previously known functions that gave the strongest separations were not transitive. The second row (in each cell) gives the reference to the Theorems where the separation result is proved. 
Although for these green cells, the bounds match that of the general functions, for some cells (with a light green color), there is a gap between the known relationships and best-known separations. 
\item[(3)] Cells with a gray background are those for which \emph{transitive functions} can be constructed in a similar but generalized fashion of our construction to demonstrate separations that match the best-known separations for general functions. 
\item[(4)] In the cells with a white background, the best-known examples for the corresponding separation were already \emph{transitive functions}. For these cells, the second row either contains the function that demonstrates the separation or is a reference to the paper where the separation was proved. So for these cells, the separations for transitive functions matched the current best-known separations for general functions. Note that for some of these cells, the bounds are not tight for general functions. 
\item[(5)] Cells with a yellow background are those where the best-known separations for transitive functions do not match the best-known separations for general functions.

\end{tablenotes}
\end{threeparttable}
}

\end{table}

\section{Basic definitions and notations}
\label{sec: main text prilims}
\subsection{Complexity measures of Boolean functions}
\label{sec: complexity measures of Boolean functions}

In this work, the relation between various complexity measures of Boolean functions is extensively studied. We refer the reader to the survey~\cite{BW} for an introduction to 
the complexity of Boolean functions and complexity measures.
Several additional complexity measures and their relations among each other can also be found in~\cite{DHT17} and~\cite{ABK+}.
Similar to the above references, we define several complexity measures of Boolean functions that are relevant to us.

We refer to~\cite{BW} and~\cite{ABBL+17} for definitions of the deterministic query model, randomized query model, and quantum query model for Boolean functions.

\begin{defi}[Deterministic query complexity]
\label{defi: dqc}
The deterministic query complexity of $f:\zone^n \to \zone$, denoted by $\dqc(f)$, is the worst-case cost of the best deterministic query algorithm computing $f$.
\end{defi}

\begin{defi}[Randomized query complexity]
\label{defi: rqc}
The randomized query complexity of $f:\zone^n \to \zone$, denoted by $\rqc(f)$, is the worst-case cost of the best-randomized query algorithm computing $f$ to error at most $1/3$.
\end{defi}

\begin{defi}[Zero-error randomized query complexity]
\label{defi: roc}
The zero-error randomized query complexity of $f:\zone^n \to \zone$, denoted by $\roc(f)$, is the minimum worst-case expected cost of a (Las Vegas) randomized query algorithm that computes $f$ to zero-error, that is, on every input $x$, if the algorithm gives output then it should give the correct answer $f(x)$ with probability $1$. Otherwise, it returns `undefined'.  
\end{defi}

\begin{defi}[Quantum query complexity]
\label{defi: qqc}
The quantum query complexity of $f:\zone^n \to \zone$, denoted by $\qqc(f)$, is the worst-case cost of the best quantum query algorithm computing $f$ to error at most $1/3$.
\end{defi}

\begin{defi}[Exact quantum query complexity]
\label{defi: qec}
The exact quantum query complexity of $f:\zone^n \to \zone$, denoted by $\qec(f)$, is
the minimum number of queries made by a quantum algorithm that outputs $f(x)$ on every input $x \in \zone^n$ with probability $1$.
\end{defi}

Next, we define the notion of partial assignment that will be used to define several other complexity measures.

\begin{defi}[Partial assignment]
\label{defi :partial assignemtn}
A partial assignment is a function $p:S \to \zone$ where $S \subseteq [n]$ and the size of $p$ is $|S|$.
For $x \in \zone^n$ we say $p \subseteq x$ if
$x$ is an extension of $p$, that is the restriction of $x$ to $S$ denoted
$x|_S = p$.
\end{defi}

\begin{defi}[Certificate complexity]
\label{defi: certificate complexity}
A $1$-certificate is a partial assignment that forces the value of the function to $1$ and similarly, a $0$-certificate is a partial assignment that forces the value of the function to $0$. The certificate complexity of a function $f$ on $x$, denoted as $\cert(x, f)$, is the size of the
smallest $f(x)$-certificate that can be extended to $x$.

Also, define $0$-certificate of $f$ as
$\cert^{0}(f) = \max\{\cert(f,x):x\in \zone^n, f(x)=0\}$ and $1$-certificate of $f$ as
$\cert^{1}(f) = \max\{\cert(f,x):x\in \zone^n, f(x)=1\}$. Finally, define the certificate complexity of $f:\zone^n \to \zone$, denoted by $\cert(f)$, to be $\max\{\cert^{0}(f), \cert^{1}(f)\}$.
\end{defi}

\begin{defi}[Unambiguous certificate complexity]
For any Boolean function $f:\zone^n\to\zone$, a set of partial assignments $U$ is said to form
an unambiguous collection of $0$-certificates for $f$ if

\begin{enumerate}
    \item Each partial assignment in $U$ is a $0$-certificate
    (with respect to $f$)
    \item For each $x\in f^{-1}(0)$, there is some $p\in U$
    with $p\subseteq x$
    \item No two partial assignments in $U$ are consistent.
\end{enumerate}

We then define $\UC_{0}(f)$ to be the minimum value
of $\max_{p\in U} |p|$ over all choices of
such collections $U$.
We define $\UC_{1}(f)$ analogously, and set
$\UC(f)= \max\{\UC_{0}(f),\UC_{1}(f)\}$.
We also define the one-sided version,
$\UC_{\min}(f)= \min\{\UC_{0}(f),\UC_{1}(f)\}$.
\end{defi}

Next, we define randomized certificate complexity (see~\cite{Aar08}).

\begin{defi}[Randomized certificate complexity]
\label{defi: RC}
A randomized verifier for input $x$ is a randomized algorithm that, on input $y$ in the domain of $f$ accepts with probability $1$ if $y=x$, and rejects with probability at least $\frac{1}{2}$ if $f(y) \neq f (x)$. If $y \neq x$ but
$f (y ) =f (x)$, the acceptance probability can be arbitrary. The Randomized certificate complexity of $f$ on input $x$ is denoted by $\mathsf{RC}(f, x)$, is the minimum expected number of queries used by a randomized verifier for $x$. The randomized certificate complexity of $f$, denoted by $\mathsf{RC}(f)$, is defined as $\max\{\mathsf{RC}(f,x): x \in \zone^n\}$.
\end{defi}

\begin{defi}[Block sensitivity]
\label{defi: Block sensitivity}
The \textit{block sensitivity} $\mathsf{bs}(f,x)$ of a function $f:\zone^n \to \zone$ on an input $x$
is the maximum number of disjoint subsets $B_1, B_2, \dots , B_r$ of
$[n]$ such that for all $j$, $f(x) \neq f(x^{B_j})$, where $x^{B_j} \in \zone^n$ is the input obtained by flipping the bits of $x$ in the coordinates in $B_j$. The \textit{block sensitivity} of $f$, denoted by $\mathsf{bs}(f)$, is $ \max \{ \mathsf{bs}\{(f,x)\}: x \in \zone^n\}$.
\end{defi}

\begin{defi}[Sensitivity]
\label{defi: sensitivity}
The \textit{sensitivity} of $f$  on an input $x$ is
defined as the number of bits on which the function is sensitive:
$\mathsf{s}(f,x) = |\{i : f(x) \neq f(x^i)\}|$. We define the sensitivity of $f$ as $\mathsf{s}(f)=\max\{\mathsf{s}(f,x):x\in \zone^n\}$

We also define $0$-sensitivity of $f$ as
$\mathsf{s}^{0}(f) = \max\{\mathsf{s}(f,x):x\in \zone^n, f(x)=0\}$, and $1$-sensitivity of $f$ as
$\mathsf{s}^{1}(f) = \max\{\mathsf{s}(f,x):x\in \zone^n, f(x)=1\}$.
\end{defi}

Here we are defining \emph{Spectral sensitivity} from~\cite{ABK+}. For more details about \emph{Spectral sensitivity} and its updated relationship with other complexity measures, we refer~\cite{ABK+}.

\begin{defi}[Spectral sensitivity]
\label{defi: lambda}
Let $f:\zone^n \to \zone$ be a Boolean function.
	The \emph{sensitivity graph} of $f$, $G_f = (V,E)$ is a subgraph of the Boolean hypercube, where $V= \zone^n$, and $E = \{(x,x\oplus e_i)\in V\times V: i\in [n], f(x)\neq f(x\oplus e_i)\}$, where $x\oplus e_i \in V$ is obtained by flipping the $i$th bit of $x$. That is, $E$ is the set of edges between neighbors on the hypercube that have different $f$ values. Let $A_f$ be the adjacency matrix of the graph $G_f$.  We define the spectral sensitivity of $f$ as the largest eigenvalue of $A_f$.
\end{defi}

\begin{defi}[Degree]
\label{defi: degree}
A polynomial $p : \mathbb{R}^n \to \mathbb{R}$ represents $f: \zone^n \to \zone$ if for all $x \in \zone^n$, $p(x) = f(x)$. The degree of a Boolean function $f$, denoted by $\mathsf{deg}(f)$, is the degree of the unique multilinear polynomial that represents $f$.
\end{defi}

\begin{defi}[Approximate degree]
\label{defi: approx degree}
A polynomial $p : \mathbb{R}^n \to \mathbb{R}$ approximately represents a function $f: \zone^n \to \zone$ if for all $x \in \zone^n$, $|p(x) - f(x)| \leq \frac{1}{3}$. The approximate degree of a Boolean function $f$, denoted by $\widetilde{\mathsf{deg}}(f)$, is the minimum degree of a polynomial that approximately represents $f$.
\end{defi}

The following is a known relation between degree and one-sided unambiguous certificate complexity measure (\cite{DHT17}).

\begin{observation}[\cite{DHT17}]
\label{obs:UC min vs deg}
For any Boolean function $f$, $\mathsf{UC_{min}}(f) \geq \mathsf{deg}(f)$.
\end{observation}

Next, we define the composition of two Boolean functions.

\begin{defi}[Composition of functions]
\label{defi: composition}
Let $f: \zone^{n} \to \zone$ and $g: \zone^m \to \zone^k$ be two functions. Then \emph{composition of $f$ and $g$}, denoted by $f \circ g: \zone^{\frac{n}{k}\cdot m} \to \zone$, is defined to be a function on $\frac{n}{k}m$ bits such that on input $x = (x_1, \dots, x_n) \in \zone^{\frac{n}{k}\cdot m}$, where each $x_i \in \zone^m$, $f \circ g(x_1, \dots, x_n) = f(g(x_1), \dots, g(x_n))$. We will refer $f$ as \emph{outer function} and $g$ as \emph{inner function}. Also, we will assume throughout that $k$ divides $n$ whenever $\frac{n}{k}$ appears in the input size of a function. 
\end{defi}

One often tries to understand how the complexity measure of a composed function behaves concerning the measures of the individual functions. The following folklore theorem that we will be using multiple times in our paper. 
\begin{theorem}
\label{theo: composition}
Let $f: \zone^n \to \zone$ and $g: \zone^m \to \zone^k$ be two Boolean functions and the composition $(f \circ g)$ is defined as Definition~\ref{defi: composition} then 
\begin{enumerate}
    \item $\dqc(f \circ g) =  \Omega(\dqc(f)/k)$, assuming $g$ is an onto function.
    \item $\qqc(f \circ g) =  \Omega(\qqc(f)/k)$, assuming $g$ is onto.
    \item $\roc(f\circ g) = O(\roc(f)\cdot m)$ and if $g$ is onto then $\roc(f\circ g) = \Omega(\roc(f)/ k)$
    \item $\rqc(f\circ g) = O(\rqc(f)\cdot m)$ and if $g$ is onto then $\rac(f\circ g) = \Omega(\rac(f)/ k)$.
    \item $\deg(f\circ g) = O(\deg(f)\cdot m)$
    \item $\adeg(f\circ g) = O(\adeg(f)\cdot m)$.
\end{enumerate}
\end{theorem}

\begin{proof}
Here we are giving a proof for (1). The proof of other lower bounds follow from a similar argument, while the upper-bounds are straight forward.

The input to $f$ in $f \circ g$ can be seen as an $n$-bit string divided into $m = n/k$ blocks each of size $k$. Since $g$ is onto, all possible strings in $\{0,1\}^n$ as possible inputs to $f$. In order to compute $f$ correctly on all inputs, at least $\dqc(f)$ may inputs to $f$ bits must be queried. Since one query to $f \circ g$ reveals at most $k$ bits of the input to $f$, we get a lower bound of $\Omega(\dqc(f)/k)$.

\end{proof}

If the inner function $g$ is Boolean valued then we can obtain some tighter results for the composed functions. 

\begin{theorem}
\label{theo: tight composition}
Let $f: \zone^n \to \zone$ and $g: \zone^m \to \zone$ be two Boolean functions then 
\begin{enumerate}
    \item (\cite{Tal13, Montanaro14}) $\dqc(f \circ g) =  \Theta(\dqc(f)\cdot \dqc(g))$.
    \item (\cite{Reichardt11,LMR+11,Kimmel13}) $\qqc(f \circ g) =  \Theta(\qqc(f)\cdot \qqc(g))$.
    \item (folklore) $\deg(f\circ g) = \Theta(\deg(f)\cdot \deg(g))$.
\end{enumerate}
\end{theorem}

Over the years a number of interesting Boolean functions have been constructed to demonstrate differences between various measures of Boolean functions. Some of the functions have been referred to in the Table~\ref{table: main_table}. We describe the various functions in the Subsection~\ref{sec:BooleanFns}.

\subsection{Some Boolean functions and their properties}\label{sec:BooleanFns}
In this section, we define some standard functions that are either mentioned in Table~\ref{table: main_table} or used somewhere in the paper. We also state some of the properties that we need for our proofs. We start by defining some basic Boolean functions.

\begin{defi}
\label{defi: parity}
Define $\PARITY: \zone^n \to \zone$ to be the $\PARITY(x_1,\dots, x_n) = \sum x_i \text{mod}~2$. We use the notation $\oplus$ to denote $\PARITY$.
\end{defi}

\begin{defi}
\label{defi: AND}
Define $\AND: \zone^n \to \zone$ to be the $\AND(x_1,\dots, x_n) = 0$ if and only if there exists an $i \in [n]$ such that $x_i = 0$. We use the notation $\wedge$ to denote $\AND$.
\end{defi}

\begin{defi}
\label{defi: OR}
Define $\OR: \zone^n \to \zone$ to be the $\OR(x_1,\dots, x_n) = 1$ if and only if there exists an $i \in [n]$ such that $x_i = 1$. We use the notation $\vee$ to denote $\OR$.
\end{defi}

\begin{defi}
\label{defi: MAJORITY}
Define $\MAJORITY: \zone^n \to \zone$ as $\MAJORITY(x) =1$ if and only if $|x| > \frac{n}{2}$.
\end{defi}

We need the following definition of composing iteratively with itself. 

\begin{defi}[Iterative composition of a function]
\label{defi: iterative composition}
Let $f: \zone^n \to \zone$ a Boolean function. For $d \in \N$ we define the function $f^d : \zone^{n^d} \to \zone$ as follows: if $d = 1$ then $f^d = f$, otherwise
\begin{align*}
    f^d(x_1, \dots, x_{n^d}) = 
    f\left(f^{d-1}(x_1, \dots, x_{n^{d-1}}), \dots, f^{d-1}(x_{n^{d}-n+1}, \dots, x_{n^{d}})\right).
\end{align*}
\end{defi}

\begin{defi}
\label{defi: NAND tree}
For $d \in \N$ define $\NAND$-tree of depth $d$ as $\NAND^d$ where $\NAND:\zone^2 \to \zone$ is defined as:
$\NAND(x_1, x_2) = 0$ if and only if $x_1 \neq x_2$. We use the notation $\widetilde{\wedge}$-tree to denote $\NAND$-tree.
\end{defi}

Now we will define a function that gives a quadratic separation between sensitivity and block sensitivity.

\begin{defi}[Rubinstein’s function (\cite{Rub95})]
\label{defi: Rubinstein’s function}
Let $g: \zone^k \to \zone$ be such that $g(x) =1$ iff $x$ contains two consecutive ones and the rest of the bits are $0$. The Rubinstein’s function, denoted by $\RUB: \zone^{k^2} \to \zone$ is defined to be $\RUB = \OR_k \circ g$.
\end{defi}

\begin{theorem}(\cite{Rub95})
\label{fact: quadratic separation s vs bs for Rubinstein’s function}
For the Rubinstein’s function in Definition~\ref{defi: Rubinstein’s function} $\mathsf{s}(\RUB) = k$ and $\mathsf{bs}(\RUB) = k^2/2$. Thus $\RUB$ witnesses a quadratic gap between sensitivity and block sensitivity.
\end{theorem}

\cite{NS94} first introduced a function whose $\deg$ is significantly smaller than $\mathsf{s}$ or $\mathsf{bs}$. This appears in the footnote in~\cite{NW94} that E. Kushilevitz also introduced a similar function with $6$ variables which gives a slightly better gap between $\mathsf{s}$ and $\deg$.  Later Ambainis computed $\qec$ of that function and gave a separation between $\qec$ and $\dqc$~\cite{Ambainis16}. This function is fully sensitive at all zero input, consequently, this gives a separation between $\qec$ and $\mathsf{s}$. 
\begin{defi}[\cite{NS94}]
\label{defi: Nisan's function}
Define $\NW$ as follows:
\begin{align*}
    \NW(x_1, x_2, x_3) =
    \begin{cases}
    1  \text{ iff } x_i \neq x_j \text{ for some  } i, j \in \{1, 2, 3\} \\
    0 \text{ otherwise}.
    \end{cases}
\end{align*}
Now define the $d$-th iteration $\NW^d$ on $(x_1, x_2, \dots, x_{3^d})$ as Definition~\ref{defi: iterative composition} where $d \in \N$.
\end{defi}

\begin{defi}[Kushilevitz’s function]
\label{defi:Kushilevitz’s function}
Define $\ku$ as follows:
\begin{align*}
    \ku(x_1, x_2, x_3,x_4,x_5,x_6) =
    \Sigma x_i + \Sigma x_iy_i + x_1x_3x_4 + x_1x_4x_5 +x_1x_2x_5 + x_2x_3x_4 + x_2x_3x_5 + x_1x_2x_6\\ +x_1x_3x_6 + x_2x_4x_6 + x_3x_5x_6 + x_4x_5x_6.
\end{align*}
Now define the $d$-th iteration $\ku^d$ on $(x_1, x_2, \dots, x_{6^d})$ as Definition~\ref{defi: iterative composition} where $d \in \N$.
\end{defi}

Next, we will describe two examples that were introduced in~\cite{GSS16} and give the separation between $\cert$ vs. $\mathsf{RC}$ and $\mathsf{RC}$ vs. $\mathsf{bs}$ respectively. They also have introduced some new complexity measures for the iterative version of a function and how to use them to get the critical measure between two complexity measures. For more details we refer to~\cite{GSS16}. 

\begin{defi}[\cite{GSS16}]
\label{defi: C vs RC}
Let $n$ be an even perfect square, let $k = 2 \sqrt{n}$ and $d= \sqrt{n}$. Divide the $n$ indices of the input into $n/k$ disjoint blocks. 
Define $\GSS_1: \zone^n \to \zone$ as follows: $\GSS_1(x) = 1$ if and only if $|x| \geq d$ and all the $1$'s in $x$ are in a single block. Define $\GSS_1^t$ with $\GSS_1^1 = \GSS_1$ and $\GSS_1^{t} = \GSS_1^{t-1}\circ \GSS_1^1$.
\end{defi}

\begin{defi}[\cite{GSS16}]
\label{defi: RC vs bs}
Define $\GSS_2: \zone^n \to \zone$ where $n$ is of the form $\binom{t}{2}$ for some integer $t$. Identify the input bits of $\GSS_2$ with the edges of the complete graph $K_t$. An input $x \in \zone^n$ induces a subgraph of $K_t$
consisting of edges assigned $1$ by $x$.
The function $\GSS_2(x)$ is defined to
be $1$ iff the subgraph induces by $x$ has a star graph.
\end{defi}

\begin{defi}
\label{defi: k-sum}
For $\Sigma= [n^k]$ the function $\ksum: \Sigma^n \to \zone$ is defined as follows: on input $x_1,x_2, \dots, x_n \in \Sigma$, if there exists $k$ element $x_{i_1}, \dots, x_{i_k}$, $i_1, \dots, i_k \in [n]$, that sums to $0~(\text{mod}~|\Sigma|)$ then output $1$, otherwise output $0$.
\end{defi}

\begin{theorem}[\cite{ABK16, BSK13}]
\label{theo: ksum qcc lower bound}
For the function $\ksum: \Sigma^n \to \zone$, if $|\Sigma| \geq 2 \binom{n}{n}$ then 
\begin{align*}
    \qqc(\ksum) = \Omega(n^{k/(k+1)}/\sqrt{k}).
\end{align*}
\end{theorem}

Next, we will define a refined version of $\ksum$ function that was defined in \cite{ABK16}. 
\begin{defi}
\label{defi: block-k-sum}
    Define $\blockksum: \zone^n \to \zone$ as follows: split the input into blocks of size $10k\log n $ each and call a block balanced if it has an equal number of $0$s and $1$s. Let the balanced blocks represent numbers in an alphabet $\Sigma$ of size $\Omega(n^k)$. The function $\blockksum$ evaluates to $1$ if and only if there are $k$ balanced blocks whose corresponding numbers sum to $0 (\mod \Sigma)$ and all other blocks have at least as many $1$s as $0$s.
\end{defi}
Next, we define the cheat sheet version of a function from~\cite{ABK16}.
\begin{defi}
\label{defi:cheatsheet}
We define the cheat sheet version of $f$ as follows: the
input to $f_{CS}$ consist of $\log n$ inputs to $f$ , each of size $n$ , followed by $n$ blocks of bits of size $C(f)\times \log n$ each.
Let us denote the input to $f_{CS}$ as $X=(x_1 , x_2 ,\dots , x_{\log n} , Y_1 , Y_2 , . . . , Y_n )$, where $x_i$ is an input to $f$ , and the $Y_i$ are the aforementioned cells of size $C(f)\times \log n$.
The first part  $x_1 , x_2 ,\dots , x_{\log n}$ of the string, we call the input section, and the rest of the part we call as certificate section of the whole input. Define $f_{CS}: \zone^{n \times \log n + n \times (C(f) \log n)} \to \zone$ to be
$1$ if and only if the following conditions hold:
\begin{itemize}
    \item
 For all $i, x_i$ is in the domain of $f$. If this condition is satisfied, let $l$ be the positive integer
corresponding to the binary string $(f (x_1 ), f (x_2 ), . . . , f (x_{\log n} ))$.
    \item $Y_l$ certifies that all $x_i$ are in the domain of $f$ and that $l$ equals the binary string formed by
their output values, $(f (x_1 ), f (x_2 ), . . . , f (x_{\log n} ))$.
\end{itemize}
\end{defi}

Finally, we present the pointer functions and their variants introduced in \cite{ABBL+17}. They are used to demonstrate the separation between several complexity measures like \emph{deterministic query complexity}, \emph{Randomized query complexity}, \emph{Quantum query complexity} etc. These functions were originally motivated from {}~\cite{GPW18} function.  The functions that we construct for many of our theorems a composition functions whose outer function is these pointer functions, or a slight variant of these. In the next section, we present the formal definition of pointer functions. 

\subsubsection{Pointer function}
\label{subsec: Pointer function and some of its variant }
For the sake of completeness first, we will describe the ``pointer function" introduced in \cite{ABBL+17} that achieves separation between several complexity measures like \emph{Deterministic query complexity}, \emph{Randomized query complexity}, \emph{Quantum query complexity} etc. This function was originally motivated by a function in \cite{GPW18}. There are three variants of the pointer function that have some special kind of non-Boolean domain, which we call \emph{pointer matrix}. Our function is a special ``encoding'' of that non-Boolean domain such that the resulting function becomes transitive and achieves the separation between complexity measures that match the known separation between the general functions. Here we will define only the first variant of the pointer function.

\begin{defi}[Pointer matrix over $\Sigma$]
\label{defi: pointer_matrix}
For $m, n \in \N$, let $M$ be a $(m\times n)$ matrix with $m$ rows and $n$ columns. We refer to each of the $m\times n$ entries of $M$ as \emph{cells}. Each cell of the matrix is from a alphabet set $\Sigma$ where $\Sigma = \zone \times \widetilde{P} \times \widetilde{P} \times \widetilde{P}$ and $\widetilde{P} = \{ (i, j) | i \in [m], j\in [n]\} \cup \{\perp\}$. We call $\widetilde{P}$ as set of pointers where, pointers of the form $\{(i,j)| i \in [m], j\in [n] \}$ pointing to the cell $(i, j)$ and $\perp$ is the null pointer. Hence,each entry $x_{(i,j)}$ of the matrix $M$ is a $4$-tuple from $\Sigma$. The elements of the $4$-tuple we refer as \emph{value}, \emph{left pointer}, \emph{right pointer} and \emph{back pointer} respectively and denote by $\vl(x_{(i,j)})$, $\lp(x_{(i,j)})$, $\rp(x_{(i,j)})$ and $\bp(x_{(i,j)})$ respectively where $\vl \in \{0, 1\}$, $\lp,\rp,\bp \in \widetilde{P}$. We call this type of matrix as \emph{pointer matrix} and denote by $\Sigma^{n \times n}$.

A special case of the pointer-matrix, which we call \emph{$\type_1$\footnote{In \cite{ABBL+17} other variations of pointer functions were defined and have been used to give various separation results, we also give transitive pointer functions for such separations. A detailed proof can be found in the previous version of the paper~\cite{CKP21}.} pointer matrix over $\Sigma$}, is when for each cell of $M$, $\bp \in \{[n] \cup \perp\}$ that is backpointers are pointing to the columns of the matrix.

\end{defi}

 Now we will define some additional properties of the domain that we need to define the pointer function.
\begin{defi}[Pointer matrix with marked column]
\label{defi: marked column ad special element}
Let $M$ be an $m\times n$ \emph{pointer-matrix} over $\Sigma$. A column $j \in [n]$ of $M$ is defined to be a \emph{marked column} if there exists exactly one cell $(i,j)$, $i \in [m]$, in that column with entry $x_{(i,j)}$ such that $x_{(i,j)} \neq (1, \perp, \perp, \perp)$ and every other $\cl$ in that column is of the form $(1, \perp, \perp, \perp)$. The $\cl$ $(i,j)$ is defined to be the \emph{special element} of the \emph{marked column} $j$.
\end{defi}

Let $n$ be a power of $2$. Let $T$ be a rooted, directed, and balanced binary tree with $ n$ leaves and $(n-1)$ internal vertices. We will use the following notations that will be used in defining some functions formally. 

\begin{notation}
\label{defi: balanced binary tree}
Let $n$ be a power of $2$.
Let $T$ be a rooted, directed, and balanced binary tree with $n$ leaves and $(n-1)$ internal vertices. 
Labels the edges of $T$ as follows: the outgoing edges from each node are labeled by either $left$ or $right$. The leaves of the tree are labeled by the elements of $[n]$ from left to right, with each label used exactly once. For each leaf $j\in[n]$ of the tree, the path from the $root$ to the leaf $j$ defines a sequence of $left$ and $right$ of length $O(\log n)$, which we denote by $T(j)$.

When $n$ is not a power of $2$, choose the largest $k \in \N$ such that $2^k \leq n$, consider a completely balanced tree with $2^k$ leaves, and add a pair of child nodes to each $n - 2^k$ leaves starting from the left. Define $T(j)$ as before.
\end{notation}

Now we are ready to describe the \emph{Variant 1} of the pointer function.

\begin{defi}[Variant 1,~\cite{ABBL+17}]
\label{defi: variant 1}
 Let $\Sigma^{m \times n}$ be a $\type_1$ \emph{pointer matrix} where $\bp$ is a pointer of the form $\{j| j\in[n]\}$ that points to other column and $\lp$, $\rp$ are as usual points to other cell.  
 Define $\abbl_{(m,n)}: \Sigma^{m \times n} \to \zone$ on a $\type_1$ \emph{pointer matrix} such that for all $x= (x_{i, j}) \in \Sigma^{m \times n}$, the function $\abbl_{(m,n)}(x_{i,j})$ evaluates to $1$ if and only if it has a \emph{$1$- cell certificate} of the following form:

\begin{enumerate}
   \item there exists exactly one \emph{marked column} $j^{\star}$ in $M$,
   \item There is a special cell, say $(i^{\star},j^{\star})$ which we call the special element in the the \emph{marked column} $j^{\star}$ and there is a balanced binary tree $T$ rooted at the special cell,
    \item for each non-marked column $j\in [n]\setminus \{j^{\star}\}$ there exist a \emph{cell} $l_j$ such that $\vl(l_j)=0$ and $\bp(l_j)=j^{\star}$ where $l_j$ is the end of the path that starts at the
\emph{special element} and follows the pointers $\lp$ and $\rp$ as specified by the sequence
$T(j)$. $l_j$ exists for all $j\in [n]\setminus \{j^{\star}\}$ i.e. no pointer on the path is $\perp$. We refer to $l_j$ as the leaves of the tree. 
\end{enumerate}
\end{defi}

\begin{figure}[!h]
    \centering
    \includegraphics[scale=0.15]{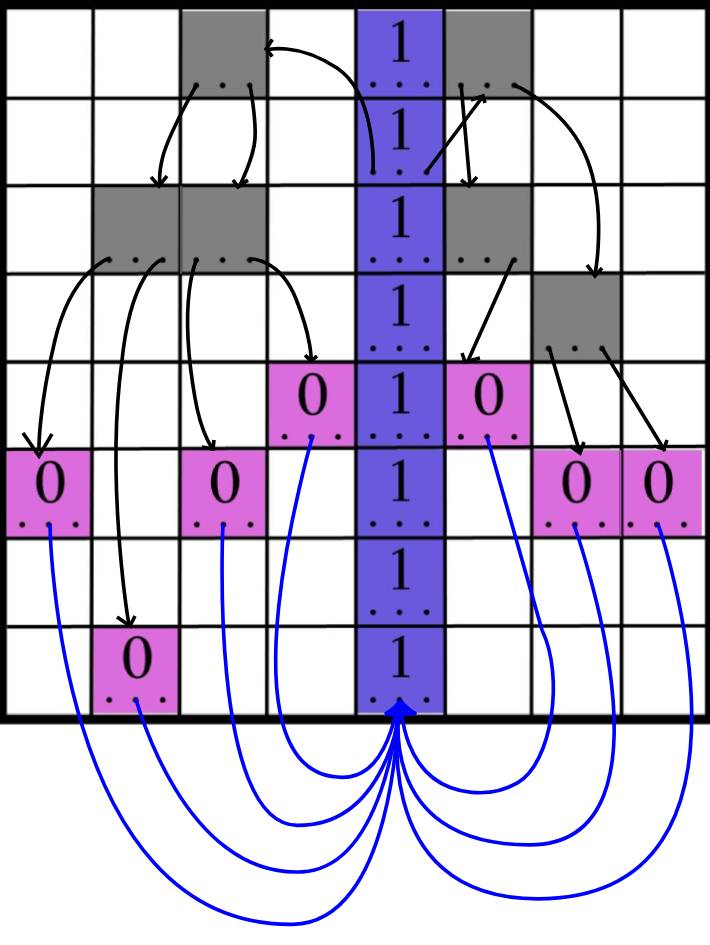}
    \caption{Example of $1$-instance of $\abbl$ function on $8 \times 8$ grid }
    \label{fig:pointer function}
\end{figure}

The above function achieves the separation between $\dqc$ vs. $\roc$ and $\dqc$ vs. $\qqc$ for $m=2n$. Here we will restate some of the results from~\cite{ABBL+17} which we will use to prove the results for our function:

\begin{theorem}[\cite{ABBL+17}]
\label{thm: Properties of variant 1}
The function $\abbl_{(m.n)}$ in Definition~\ref{defi: variant 1} satisfies  
\begin{align*}
  \dqc &= \Omega(n^2) \text{ for $m=2n$ where $m,n \in \N$}, \\
  \roc &= \widetilde{O}(m+n) \text{ for any $m,n \in \N$},\\
  \qqc &= \widetilde{O}(\sqrt{m}+\sqrt{n}) \text{ for any $m,n \in \N$}.
\end{align*}
\end{theorem}
 
Though~\cite{ABBL+17} gives the deterministic lower bound for the function $\abbl$ precisely for $2m \times m$ matrices following the same line of argument it can be proved that $\dqc( \Omega(n^2))$ holds for $n \times n$ matrices also.
For the sake of completeness, we give a proof for $n \times n$ matrices.

\begin{theorem}
\label{thm:D(abbl)}
 $\dqc(\abbl_{(n,n)}) = \Omega(n^2)$.
\end{theorem}

\textbf{Adversary Strategy for $\abbl_{(n,n)}$}:
We describe an adversary strategy that ensures that the value of the function is undetermined after $\Omega(n^2)$ queries. Assume that a deterministic query algorithm queries a cell $(i,j)$. Let $k$ be the number of queried cells in the column $j$. If $k\leq \frac{n}{2}$ adversary will return $(1,\perp,\perp,\perp)$. Otherwise adversary  will return $(0,\perp,\perp,n-k)$.

\begin{claim}
The value of the function $\abbl_{(n, n)}$ will be undetermined if there is a column with at most $n/2$ queried cells in the first $\frac{n}{2}$ columns $\{1, 2, \dots, \frac{n}{2}\}$ and at least $3n$ unqueried cells in total.
\end{claim} 

\begin{proof}
The adversary can always set the value of the function to $0$ if the conditions of the claim are satisfied.

\textbf{Adversary can also set the value of the function to $1$:} If $s\in[\frac{n}{2}]$ be the column with at most $\frac{n}{2}$ queried cell, then all the queried cells of the column are of the form $(1,\perp,\perp,\perp)$. Assign $(1,\perp,\perp,\perp)$ to the other cell and leave one cell for the $special$ $element$ $a_{p, s}$(say).

For each non-marked column $j \in [n]\setminus \{s\}$ define $l_j$ as follows: If column $j$ has one unqueried cell then assign $(0, \perp, \perp, s)$ to that cell. If all the cells of the column $j$ were already queried then the column contains a cell with $(0, \perp, \perp, s)$ by the adversary strategy. So, in either case, we can form a \emph{leave} $l_j$ in each of the non-marked columns.

Now using the cell of \emph{special element} $a_{p, s}$ construct a rooted tree of pointers isomorphic to tree $T$ as defined in Definition~\ref{defi: variant 1} such that the internal nodes we will use the other unqueried cells and assign pointers such that $l(j)$'s are the leaves of the tree and the \emph{special element} $a_{p,s}$ is the $root$ of the tree.
Finally, assign anything to the other cell. Now the function will evaluate to $1$.
 
 To carry out this construction we need at most $3n$ number of unqueried cells. Outside of the $marked$ $column$ total $n-2$ cells for the internal nodes of the tree, at most $n-1$ unqueried cell for the $leaves$ and the \emph{all $1$ unique marked column} contains total $n$ cell, so total $3n$ unqueried cell will be sufficient for our purpose.

Now there are a total $n$ number of columns and to ensure that each of the columns in $\{1, 2, \dots, \frac{n}{2} \}$ contains at least $\frac{n}{2}$ queried cell we need at least $\frac{n^2}{4}$ number of queries. Since $n^2-3n\geq \frac{n^2}{4}$ for all $n\geq 6$. Hence $D(\abbl_{(n,n)} )= \Omega(n^2)$. 
\end{proof}

Hence Theorem~\ref{thm:D(abbl)} follows.

Also~\cite{GPW18}'s function realizes quadratic separation between $\dqc$ and $\deg$ and the proof goes via $\UC_{min}$ upper bound. The function $\abbl_{(n,n)}$ exhibits the same properties corresponding to $\UC_{min}$. So, from the following observation, it follows that $\abbl_{(n,n)}$ also achieves quadratic separation between $\dqc$ and $\deg$.

\begin{observation}
 \label{obs: UC(Fnn)}
$\UC_{min}(\abbl_{(n,n)}) = O(n)$ which implies $\deg(\abbl_{(n,n)})$ is also $O(n)$ for any $n\in \N$.
\end{observation}

Another important observation that we need is the following:
 
\begin{observation}[\cite{ABBL+17}]
\label{thm: invariance of variant 1}
For any input $\Sigma^{n\times n}$ to the function $\abbl_{(n,n)}$ (in Definition~\ref{defi: variant 1}) if we permute the 
rows of the matrix using a permutation $\sigma_r$ and permute the columns of the matrix using a permutation $\sigma_c$ and we update the pointers in each of the cells of the matrix accordingly then the function value does not change. 
\end{observation}

\subsection{Some useful notations}
We use $[n]$ to denote the set $\{1,\dots, n\}$. $\{0,1\}^n$ denotes the set of all $n$-bit binary strings. For any $X\in \{0,1\}^n$ the Hamming Weight of $X$ (denoted $|X|$) will refer to the number of $1$ in $X$. $0^n$ and $1^n$ denotes all $0$'s string of $n$-bit and all $1$'s string of $n$-bit, respectively.

We denote by $\sym_n$ the set of all permutations on $[n]$. Given an element $\sigma \in \sym_n$ and a $n$-bit string $x_1,\dots,x_n \in \{0,1\}^n$ we denote by $\sigma[x_1, \dots, x_n]$ the string obtained by permuting the indices according to $\sigma$. That is $\sigma[x_1, \dots, x_n] = x_{\sigma^(1)}, \dots, x_{\sigma^(n)}$. This is also called the action of $\sigma$ on the $x_1, \dots, x_n$. 

Following are a couple of interesting elements of $\sym_n$ that will be used in this paper. 

\begin{defi}
\label{defi: flip swap}
For any $n=2k$ the $\flip$ swaps $(2i-1)$ and $2i$ for all $1\leq i\leq k$.
The permutation $\swap$ swaps $i$ with $(k+i)$, for all $1\leq i \leq k$.
That is,
\begin{align*}
    \flip = (1,2)(3,4)\dots (n-1, n) \hspace{2em} \& \hspace{2em}  \swap[x_1, \dots, x_{2k}] = x_{k+1},\dots, x_{2k},x_1 \dots, x_k. 
\end{align*}
\end{defi}

Every integer $\ell \in [n]$ has the canonical $\log n$ bit string representation. However, the number of $1$'s and $0$'s in such a representation is not the same for all $\ell \in [n]$. The following representation of $\ell \in [n]$ ensures that for all $\ell \in [n]$ the encoding has the same Hamming weight.

\begin{defi}[Balanced binary representation]
\label{defi: bb(i)}
For any $\ell \in [n]$, let $\ell_1, \dots, \ell_{\log n}$ be the binary representation of the number $\ell$ where $\ell_i \in \zone$ for all $i$. Replacing $1$ by $10$ and $0$ by $01$ in the binary representation of $\ell$, we get a $2\log n$-bit unique representation, which we call \emph{Balanced binary representation} of $\ell$ and denote as $bb(\ell)$.
\end{defi}

In this paper, all the functions considered are of the form $F:\zone^n \to \zone^k$.  By Boolean functions, we would mean a Boolean valued function that is of the form $f: \zone^n \to \zone$. 

An input to a function $F:\zone^n \to \zone^k$  is a $n$-bit string but also the input can be thought of as different objects.  For example, if the $n= NM$ then the input may be thought of as a $(N\times M)$-matrix 
with Boolean values. It may also be thought of as a $(M\times N)$-matrix.

If $\Sigma= \{0,1\}^k$ then $\Sigma^{(n\times m)}$ denotes an 
$(n\times m)$-matrix with an element of $\Sigma$ (that is, a $k$-bit string)
stored in each cell of the matrix. Note that 
  $\Sigma^{(n\times m)}$ is actually $\{0,1\}^{mnk}$. Thus, a function
$F:\Sigma^{(n\times m)}\to \{0,1\}$ is actually a Boolean function from
a $\{0,1\}^{nmk}$ to $\zone$, where we think of the input as an
$(n\times m)$-matrix over the alphabet $\Sigma$.

One particular nomenclature that we use in this paper is that of $1$-cell certificate.  
\begin{defi}[$1$-cell certificate]
\label{defi: cell certificate} Given a function $f:\Sigma^{(n\times m)} \to \{0,1\}$ (where $\Sigma= \{0,1\}^k$) the $1$-cell certificate is a partial assignment to the cells which forces the value of the function to $1$. So a $1$-cell certificate is of the form $(\Sigma \cup \{*\})^{(n\times m)}$. Note that here we assume that the contents in any cell are either empty or a proper element of $\Sigma$ (and not a partial $k$-bit string). 
\end{defi}

Another notation that is often used is the following: 

\begin{notation}
If $A \leq \sym_n$ and $B \leq \sym_m$ are groups on $[n]$ and $[m]$ then the group $A \times B$ acts on the cells on the matrix. Thus for any $(\sigma, \sigma')\in A\times B$ and a $M\in \Sigma^{(n\times m)}$ by $(\sigma, \sigma')[M]$
we would mean the permutation on the cell of $M$ according to  $(\sigma, \sigma')$ and move the contains in the cells accordingly. Note that the relative position of bits within the contents in each cell is not touched.
\end{notation}



\subsection{Transitive groups and transitive functions}

The central objects in this paper are transitive Boolean functions. We first define transitive groups. 

\begin{defi}
A group $G\leq \sym_n$ is transitive if for all $i, j\in [n]$
there exists a $\sigma\in G$ such that $\sigma(i) = j$.
\end{defi}

\begin{defi}
For $f:A^n \to \{0,1\}$ and $G\leq \sym_n$ we say $f$ is invariant under the action of $G$, if for all $\alpha_1, \dots, \alpha_n\in A$.  $$f(\alpha_1, \dots, \alpha_n) = f(\alpha_{\sigma(1)}, \dots, \alpha_{\sigma(n)}).$$
\end{defi}

The following observation proves that the composition of transitive functions is also a transitive function.

\begin{observation}
\label{obs: composition of transitive symmetric}
Let $f: \zone^n \to \zone$ and $g: \zone^m \to \zone$ be transitive functions. Then $f \circ g: \zone^{nm} \to \zone$ is also transitive.
\end{observation}

\begin{proof}
Let $T_f \subseteq \sym_n$ and $T_g \subseteq \sym_m$ be the transitive groups corresponding to $f$ and $g$, respectively. On input $x = (X_1, \dots, X_{n})$, $X_i \in \zone^m$ for $i \in [n]$, the function $f \circ g$ is invariant under the action of the group $T_f\wr T_g$ - the wreath product of the $T_f$ with $T_g$. The group $T_f\wr T_g$ acts on the input string through the following permutations:
\begin{enumerate}
    \item any permutation $\pi \in T_f$ acting on indices $\{1, \dots, n\}$ or
    \item any permutations $(\sigma_1, \dots, \sigma_n) \in (T_g)^{n}$ acting on $X_1, \dots, X_n$ i.e. $(\sigma_1, \dots, \sigma_n)$ sends $X_1, \dots, X_n$  to $\sigma_1(X_1), \dots, \sigma_n(X_n)$.
\end{enumerate}
\end{proof}

\begin{observation}
If $A \leq \sym_n$ and $B \leq \sym_m$ are transitive groups on $[n]$ and $[m]$ then the group $A \times B$ is a transitive group acting on the cells on the matrix. 
\end{observation}

There are many interesting transitive groups. The symmetric group is indeed transitive. The graph isomorphism group (that acts on the adjacency matrix - minus the diagonal -  of a graph by changing the ordering on the vertices) is transitive. The cyclic permutation over all the points in the set is a transitive group. The following is another non-trivial transitive group on $[k]$ that we will use extensively in this paper.

\begin{defi}\label{def:transitivegroup}
For any $k$ that is a power of $2$, the Binary-tree-transitive group $\bt_k$ is a subgroup of $\sym_{k}$. To describe its generating set we think of group $\bt_k$ acting on the elements $\{1, \dots, k\}$ and the elements are placed in the leaves of a balanced binary tree of depth $\log k$ - one element in each leaf. Each internal node (including the root) corresponds to an element in the generating set of $\bt_k$. The element corresponding to an internal node in the binary tree swaps the left and right sub-tree of the node.  
The permutation element corresponding to the root node is called the Root-swap as it swaps the left and right sub-tree to the root of the binary tree. 
\end{defi}


\begin{figure}[!h]
    \centering
\begin{tikzpicture}[->,>=stealth',level/.style={sibling distance = 5cm/#1,
  level distance = 1.5cm}] 
\node [arn_r] {n1}
    child{ node [arn_r] {n3} 
            child{ node [arn_x] {1} 
            }
            child{ node [arn_x] {2}
            }                            
    }
    child{ node [arn_r] {n2}
            child{ node [arn_x] {3} 
            }
            child{ node [arn_x] {4}
            }
		}
; 
\end{tikzpicture}
 \caption{Induced group actions for $\bt_4$ group }
    \label{fig:Bt_4 group}
\end{figure}

\begin{claim}
\label{obs: btk is a transitive group}
The group $\bt_k$ is a transitive group. 
\end{claim}
\begin{proof}
For any $i, j \in [k]$, we have to show that there exists a permutation $\pi \in \bt_k$ such that $\pi(i)= j$. Let us form a complete binary tree of height $\log k$ in the following way:
\begin{itemize}
    \item (Base case:) Start from the root node, and label the left and right child as $0$ and $1$ respectively. 
    \item For every node $x$, label the left and right child as $x0$ and $x1$ respectively.
\end{itemize}
Note that our complete binary tree has $k$ leaves, where each of the leaves is labeled by a binary string of the form $x_1x_2\dots x_{\log k}$, which is the binary representation of numbers in $[k]$. Similarly, any node in the tree can be labeled by a binary string $x_1 x_2 \dots x_t$, where $0\leq t\leq \log k$ and $t$ is the distance of the node from the root.

Now for any $i, j \in [k]$, let the binary representation of  $i$ be $(x_1x_2\dots x_{\log k})$ and that of $j$ be $(y_1y_2\dots y_{\log k})$. 
Now we will construct the permutation $\pi \in \bt_k$ such that $\pi(i)= j$. Without loss of generality, we can assume $i \neq j$.

Find the least positive integer $\ell\in [\log k]$ such that $x_{\ell} \neq y_{\ell}$, then go to the node labeled $x_1x_2\dots x_{\ell-1}$ and swap it's left and right child.  Let $\pi_{x_1\dots x_{\ell-1}}\in S_{k}$ be the corresponding permutation of the leaves of the tree, in other words on the set $[k]$. Note that, by definition, the permutation $\pi_{x_1\dots x_{\ell-1}}\in S_{k}$ is in $\bt_k$.
Also note that the permutation $\pi_{x_1\dots x_{\ell-1}}$ acts of the set $[k]$ as follows: 
\begin{itemize}
\item $\pi_{x_1\dots x_{\ell-1}} (z_1 \dots z_{\log k}) = z_1\dots z_{\log k}$ if $z_1\dots z_{\ell-1} \neq x_1\dots x_{\ell-1}$
\item $\pi_{x_1\dots x_{\ell-1}} (x_1 \dots x_{\ell-1}0z_{\ell+1}\dots z_{\log k}) = (x_1 \dots x_{\ell-1}1z_{\ell+1}\dots z_{\log k})$
\item $\pi_{x_1\dots x_{\ell-1}} (x_1 \dots x_{\ell-1}1z_{\ell+1}\dots z_{\log k}) = (x_1 \dots x_{\ell-1}0z_{\ell+1}\dots z_{\log k})$
\end{itemize}

Since $i = x_1\dots x_{\ell-1}x_{\ell}\dots x_{\log k}$ and $j = y_1\dots y_{\ell-1}y_{\ell}\dots y_{\log k}$ with $x_1 \dots x_{\ell-1} = y_1\dots y_{\ell-1}$ and $x_{\ell} \neq y_{\ell}$, so 
$$\pi_{x_1\dots x_{\ell-1}} (i) = y_1\dots y_{\ell-1}y_{\ell}x_{\ell+1}\dots y_{\log k}$$

So the binary representation of $\pi_{x_1\dots x_{\ell-1}} (i)$ and $j$ matching in the first $\ell$ positions which is one more than the number of positions where the binary representation of $i$ and $j$ matched. By doing this trick repeatedly, that is by applying different permutations from $\bt_k$ one after another we can map $i$ to $j$.   
\end{proof}

\begin{note}
     Note that the group $\bt_k$ is the same as an iterated wreath product $\sym_2 \wr \dots \wr \sym_2$ of depth $\log k$.

\end{note}

The following claim describes how the group $\bt_k$ acts on various encodings of integers.
Recall the balance-binary representation (Definition~\ref{defi: bb(i)}).

\begin{claim}\label{obs:bb(i)}
For all $\hat\gamma\in \bt_{2\log n}$ there is a ${\gamma}\in \sym_n$ such that for all $i, j\in [n]$, $\hat\gamma[bb(i)] = bb(j)$ iff $\gamma(i) = j$ where $2 \log n$ is a power of $2$.
\end{claim}

\begin{proof}
Recall the group $\bt_{2\log n}$: assuming that the elements of $[2\log n]$ are placed on the leaves of the binary tree of depth $\log(2\log n)$, the group $\bt_{2\log n}$
is generated by the permutations of the form ``pick a node in the binary tree of and swap the left and right sub-tree of the node''.  So it is enough to prove that for any elementary permutation $\hat\gamma$ of the form ``pick a node in the binary tree and swap the left and right sub-tree of the node'' there is a ${\gamma}\in \sym_n$ such that for all $i, j\in [n]$, $\hat\gamma[bb(i)] = bb(j)$ iff $\gamma(i) = j$.

Any node in the binary tree of depth $\log(2\log n)$ can be labeled by a $0/1$-string of length $t$, where $0\leq t\leq \log(2\log n)$ is the distance of the node from the root. We split the proof of the claim into two cases depending on the value of the $t$ - the distance from the root. 

\paragraph{If $t= \log(2\log n)$:} This is the case when the node is at the last level - just above the leaf level.  Let the node be $u$ and let $s$ be the number whose binary representation is the label of the node $u$. Let the numbers in the leaves of the tree correspond to the $bb(i)$ - the balanced binary representation of $i \in [n]$. Note that because of the balanced binary representation, the children of $u$ are
\begin{itemize}
    \item $0$ (left-child) and $1$ (right-child)  if the $s$-th bit in the binary representation of $i$ is $0$
    \item $1$ (left-child) and $0$ (right-child)  if the $s$-th bit in the binary representation of $i$ is $1$
\end{itemize}
So the permutation (corresponding to swapping the left and right sub-trees of $u$) only changes the order of $0$ and $1$ - which corresponds to flipping the $s$-th bit of the binary representation of $i$. And so in this case the $\gamma$ acting on the set $[n]$ is just collection transpositions swapping $i$ and $j$ iff the the binary representation of $i$ and $j$ are the same except for the $s$-th bit. 

So in this case for all $i, j\in [n]$, $\hat\gamma[bb(i)] = bb(j)$ iff $\gamma(i) = j$.

\paragraph{If $t < \log(2\log n)$:} Let the node be $v$. Note that in this case since the node keeps the order of the $2r-1$ and $2r$ bits unchanged (for any $1\leq r \leq \log n$), it is enough we can visualize the action by an action of swapping the left and right sub-trees of the node $v$ on the binary representation of $i$ (instead of the balance binary representation of $i$). And so we can see that the action of the permutation (corresponding to swapping the left and right sub-trees of $v$) automatically gives a permutation of the binary representations of numbers between $1$ and $n$, as was discussed in the proof of Claim~\ref{obs: btk is a transitive group}.  And hence we have for all $i, j\in [n]$, $\hat\gamma[bb(i)] = bb(j)$ iff $\gamma(i) = j$.

\end{proof}

Now let us consider another encoding that we will be using for the set of rows and columns of a matrix. 

\begin{defi}\label{def:bpencoding}
Given a set $R$ of $n$ rows $r_1, \dots, r_n$ and a set $C$ of $n$ columns $c_1, \dots, c_n$ we define the \textit{balanced-pointer-encoding function} 
$\mathcal{E}: (R\times \{0\})\cup (\{0\}\times C)  \to \{0,1\}^{4\log n},$
as follows: 
\begin{equation*}
\mathcal{E}(r_i, 0)  =  bb(i)\cdot 0^{2\log n}, \mbox{ and, }  \mathcal{E}(0, c_j)  =  0^{2\log n}\cdot bb(j).  
\end{equation*}
\end{defi}

Note that Claim~\ref{cl:trick} directly follows from Claim~\ref{obs:bb(i)}.

\begin{claim}\label{cl:trick}
Let $R$ be a set of $n$ rows $r_1, \dots, r_n$ and $C$ be a set of $n$ columns $c_1, \dots, c_n$ and consider the \textit{balanced-pointer-encoding function} $\mathcal{E}: (R\times \{0\})\cup (\{0\}\times C) \to \{0,1\}^{4\log n}$.
For any elementary permutation $\hat{\sigma}$ in $\bt_{4\log n}$ (other than the Root-swap) there is a $\sigma\in \sym_n$ such that for any $(r_i, c_j) \in (R\times \{0\})\cup (\{0\}\times C)$
$$\hat{\sigma}[\mathcal{E}(r_i, c_j)] = \mathcal{E}(r_{\sigma(i)}, c_{\sigma(j)}),$$ where we assume $r_0 = c_0 = 0$ and any permutation of in $\sym_n$ sends $0$ to $0$.

If $\hat{\sigma}$ is the root-swap then for any $(r_i, c_j) \in (R\times \{0\})\cup (\{0\}\times C)$ $$\hat{\sigma}[\mathcal{E}(r_i, c_j)] = \swap(\mathcal{E}(r_i, c_j)) = \mathcal{E}(c_j, r_i).$$
\end{claim}

\section{High-level description of our techniques}
\label{sec: our technique}
\textit{Pointer functions} are defined over a special domain called \emph{pointer matrix}, which is a $m\times n$ grid matrix. Each cell of the matrix contains some labels and some pointers that point either to some other cell or to a row or column \footnote{We naturally think of a pointer pointing to a cell as two pointers - one pointing to the row and the other to the column.}.
 For more details, refer to Appendix~\ref{subsec: Pointer function and some of its variant }.  
As described in \cite{GPW18}, the high-level idea of pointer functions is the usage of pointers to make certificates unambiguous without increasing the input size significantly. This technique turns out to be very useful in giving separations between various complexity measures as we see in \cite{MS}, \cite{GPJW18}, and \cite{ABBL+17}. Note that \emph{Pointer function} contains non-Boolean inputs alphabets. Consequently, our transitive function also involves non-Boolean input alphabets. 

Now we want to produce a new function that possesses all the properties of pointer functions, along with the additional property of being transitive. To do so, first, we will encode the labels so that we can permute the bits (by a suitable transitive group) while keeping the structure of unambiguous certificates intact so that the function value remains invariant. One such natural technique would be to encode
the contents of each cell in such a way that allows us to permute the bits of the contents of each cell using a transitive group and permute the cells among each other using another transitive group, and doing all of these while ensuring the unambiguous certificates remains intact
\footnote{Here, we use the word ``encode" since we can view the function defined only over codewords, and when the input is not a codeword, then it evaluates to $0$. In our setting, since we are trying to preserve the one-certificates, the codewords are those strings where the unambiguous certificate is encoded correctly. At the same time, we must point out that the encoding of an unambiguous certificate is not necessarily unique.}.
This approach has a significant challenge: namely how to encode the pointers.

The information stored in each cell (other than the pointers) can be encoded using fixed logarithmic length strings of different Hamming weights so that even if the strings are permuted and/or the bits in each string are permuted, the content can be ``decoded". Unfortunately, this can only be done when the cell's contents have a constant amount of information - which is the case for pointer functions (except for the pointers). Since the pointers in the cell are strings of size $O(\log n)$ (as they are pointers to other columns or rows) if we want to use the similar Hamming weight trick, the size of the encoding string would need to be polynomial in $O(n)$. That would increase the size of the input compared to the unambiguous certificate. This would not give us tight separation results.

Also, there are three more issues concerning the encodings of pointers: 
\begin{itemize}
    \item As we permute the cells of the matrix according to some transitive group, the pointers within each cell need to be appropriately changed.  In other words, when we move some cell's content to some other cell, the pointers pointing to the previous cell should point to the current cell now.
     
    \item If a pointer is encoded using a certain $t$-bit string, different permutations of bits of the encoded pointer can only generate a subset of all $t$-bit strings. 
    
    {\em For example: if we encode a pointer using a string of Hamming weight 10 then however we permute the bits of the string, the pointer can at most be modified to point to cells (or rows or columns) the encoding of whose pointers also have Hamming weight 10. (The main issue is that permuting the bits of a string cannot change the Hamming weight of a string).}
    
    The encoding of all the pointers should have the same Hamming weight.

    \item The encoding of the pointers has to be transitive. That is, we should be able to permute the bits of the encodings of the pointer using a transitive group in such a way that either the pointer value does not change or as soon as the pointer values change, the cells get permuted accordingly - kind of like an ``entanglement". 
    
\end{itemize}

The above three problems are somewhat connected. Our first innovative idea is to use \emph{binary balance representation} (Definition~\ref{defi: bb(i)}) to represent the pointers. This way, we take care of the second issue. For the first and third issues, we define the transitive group - both the group acting on the contents of the cells (and hence on the encoding of the pointers) and the group acting on the cells itself - in an ``entangled" manner. For this, we induce a group action acting on the nodes of a \emph{balanced binary tree} and generate a transitive subgroup in $S_{n}$ and $S_{2 \log n}$ with the same action which will serve our purpose (Definition~\ref{def:transitivegroup}, Claim~\ref{obs:bb(i)}).  
This helps us to permute the rows (or columns) using a permutation while updating the encoding of the pointers accordingly.

By Claim~\ref{obs:bb(i)}, for every allowed permutation $\sigma$ acting on the rows (or columns), there is a unique $\hat{\sigma}$ acting on the encodings of the pointers in each of the cells such that the pointers are updated according to $\sigma$. This still has a delicate problem. Namely, each pointer is either pointing to a row or column, but the permutation $\hat{\sigma}$ has no way to understand whether the encoding on which it is being applied points to a row or column. To tackle this problem, we think of the set of rows and columns as a single set. All of them are encoded by a string of size (say) $2t$, where for the rows, the second half of the encoding is all $0$ while the columns have the first $t$ bits all $0$. This is the encoding described in Definition~\ref{def:bpencoding} using binary balanced representation.
However, this adds another delicate issue about permuting between the first $t$ bits of the encoding and the second $t$ bits.

To tackle this problem, we modify the original function appropriately. We define a slightly modified version of existing pointer functions called $\mabbl$.  
This finally helps us obtain our ``transitive pointer function," which has almost the same complexities as the original pointer function. 

We have so far only described the high-level technique to make the 1st variation of pointer functions (Definition~\ref{defi: variant 1}) transitive where there is the same number of rows and columns. The further variations need a more delicate handling of the encoding and the transitive groups - though the central idea is similar.

\section{Separations between deterministic query complexity and some other complexity measures}
\label{sec: variant 1}

\subsection{Transitive pointer function $F_{\ref{thm:R_0vsD}}$ for Theorem~\ref{thm:R_0vsD}}
\label{subsec: F1.1}
Our function $F_{\ref{thm:R_0vsD}}:\Gamma^{n \times n} \to \{0,1\}$ is a composition of two functions - an outer function $\mabbl_{(n,n)}: \bar\Sigma^{n\times n} \to \{0,1\}$ and an inner function $\De : \Gamma \to \bar\Sigma$. We will set $\Gamma$ to be  $\{0,1\}^{96 \log n}$.

The outer function is a modified version of the $\abbl_{(n,n)}$ - pointer function described in \cite{ABBL+17}  (see Definition~\ref{defi: variant 1} for a description). The function $\abbl_{(n,n)}$ takes as input a $(n\times n)$-matrix whose entries are from a set $\Sigma$, and the function evaluates to $1$ if a certain kind of $1$-cell-certificate exists.  Let us define a slightly modified function 
$\mabbl_{(n,n)}: \bar{\Sigma}^{n\times n} \to \{0,1\}$ where  $\bar\Sigma = \Sigma\times \{\vdash, \dashv\}$. 
We can think of an input $A \in \bar{\Sigma}^{n\times n}$ as a pair of matrices $B\in {\Sigma}^{n\times n}$ and $C\in {\{\vdash, \dashv\}}^{n\times n}$. The function $\mabbl_{(n,n)}$ is defined as 
$$\mabbl_{(n,n)}(A) = 1 \mbox{ iff } \left\{
    \begin{array}{rl}
        \mbox{ Either, }   \mbox{ (i) } &  \abbl_{(n,n)}(B) = 1, \mbox{ and, all the cells in the }\\
         &  \mbox{$1$-cell-certificate have $\vdash$ in the corresponding cells in $C$}  \\
        \mbox{ Or, }  \mbox{ (ii) } & \abbl_{(n,n)}(B^T) = 1, \mbox{ and, all the cells in the}\\
        & \mbox{$1$-cell-certificate have $\dashv$ in the corresponding cells in $C^T$}  \\
    \end{array}
\right.$$

Note that both the two conditions (i) and (ii) cannot be satisfied simultaneously.  From this it is easy to verify that the function $\mabbl_{(n,n)}$
has all the properties as $\abbl_{(n,n)}$ as described in Theorem~\ref{thm: Properties of variant 1}.

The inner function $\De$ (we call it a decoding function) is a function from $\Gamma$ to $\bar\Sigma$, where $\Gamma = 96 \log n$. Thus our final function is $$F_{\ref{thm:R_0vsD}} := \left(\mabbl_{(n, n)}\circ \De\right) : \Gamma^{n \times n} \to \{0, 1\}.$$

\subsubsection{Inner function $\De$}
\label{sec: encoding scheme of variant 1}

The input to $\abbl_{(n,n)}$ is a $\type_1$ pointer matrix $\Sigma^{n\times n}$. Each cell of a $\type_1$ pointer matrix contains a 4-tuple of the form $(\vl, \lp, \rp, \bp)$ where $\vl$ is either  $0$ or $1$ and $\lp,\rp$ are pointers to the other cells of the matrix and $\bp$ is a pointer to a column of the matrix (or can be a null pointer also).
Hence, $\Sigma = \{0,1\}\times [n]^2 \times [n]^2 \times [n]$.
For the function $\abbl_{(n,n)}$
it was assumed (in \cite{ABBL+17}) that the elements of $\Sigma$ is encoded as a $k$-length\footnote{For the canonical encoding $k = (1 + 5\log n)$ was sufficient} binary string in a canonical way.

The main insight for our function $F_{\ref{thm:R_0vsD}} := \left(\mabbl_{(n, n)}\circ \De\right)$ is that we want to maintain the basic structure of the function $\abbl_{(n,n)}$ (or rather of $\mabbl_{(n, n)}$) but at the same time we want to encode the $\bar\Sigma=\Sigma\times \{\vdash, \dashv\}$ in such a way that the resulting function becomes transitive. To achieve this, instead of having a unique way of encoding an element in $\bar\Sigma$ we produce a number of possible encodings\footnote{We use the term ``encoding'' a bit loosely in this context as technically an encoding means a unique encoding. What we actually mean is the pre-images of the function $\De$.} for any element in $\bar\Sigma$. 
The inner function $\De$ is therefore a decoding algorithm that given any proper encoding of an element in $\bar\Sigma$ will be able to decode it back.

For ease of understanding we start by describing the possible ``encodings'' of $\bar\Sigma$, that is by describing the pre-images of any element of $\bar\Sigma$ in the function $\De$.

\vspace{0.8em}

\noindent\textbf{``Encodings'' of the content of a cell in $\bar\Sigma^{n\times n}$ :}

\vspace{0.8em}

We will encode any element of $\bar\Sigma$ using a string of size $96 \log n $ bits. Recall that, an element in $\bar\Sigma$ is of the form $(V, (r_{L}, c_{L}), (r_{R}, c_{R}), (c_{B}), T)$, where $V$ is the Boolean value, $(r_L, c_L)$, $(r_R, c_R)$ and $c_B$ are the left pointer, right pointers and bottom pointer respectively and $T$ take the value $\vdash$ or $\dashv$.  The overall summary of the encoding is as follows:

\begin{itemize}
    \item \textbf{Parts:} We will think of the tuple as 7 objects, namely $V$, $r_{L}$, $c_{L}$, $r_{R}$, $c_{R}$, $c_{B}$ and $T$.  We will use $16\log n$ bits to encode each of the first 6 objects. The value of $T$ will be encoded cleverly. So the encoding of any element of $\bar\Sigma$ contains 6 \emph{parts} - each a binary string of length $16\log n$. 
    \item \textbf{Blocks:} Each of 6 \emph{parts} will be further broken into 4 \emph{blocks} of equal length of $4\log n$. One of the blocks will be a special block called the ``encoding block''. 
\end{itemize}

Now we explain, for a tuple  $(V, (r_{L}, c_{L}), (r_{R}, c_{R}), (c_{B}), T)$ what is the 4 blocks in each part. 
We will start by describing a ``standard-form'' encoding of a tuple $(V, (r_{L}, c_{L}), (r_{R}, c_{R}), (c_{B}), T)$ where $T = \vdash$. Then we will extend it to describe the standard for the encoding of $(V, (r_{L}, c_{L}), (r_{R}, c_{R}), (c_{B}), T)$ where $T = \dashv$.
Finally, we will explain all other valid encodings of a tuple  $(V, (r_{L}, c_{L}), (r_{R}, c_{R}), (c_{B}), T)$ by describing all the allowed permutations on the bits of the encoding.

\begin{table}[]
\scalebox{0.9}{
    \centering
    \begin{tabular}{| c || c | c | c | c ||c |}
    \hline
    $\dots$ & $B_1$ ``encoding''-block  & $B_2$ &  $B_3$  & $B_4$ & Hamming weight \\
    \hline
   $P1$ & $\ell_1\ell_2$, where $|\ell_1| = 2 \log n,$ and  & $ 4\log n$ & $2 \log n + 1$& $2 \log n +2 $ & $12 \log n +2 - V$\\
          &$|\ell_2| = 2\log n-1 - V$  & & & &\\
   \hline
  $P2$& $\mathcal{E}(r_L, 0)$ & $2 \log n +3$ & $2 \log n +1$ &$2 \log n +2$ & $7 \log n +6$\\
    \hline
   $P3$  & $\mathcal{E}(0, c_L)$ &$2 \log n +4$ & $2 \log n +1$  &$2 \log n +2$ & $7 \log n +7$\\
    \hline
  $P4$  &  $\mathcal{E}(r_R, 0)$ & $2 \log n +5$& $2 \log n +1$  &$2 \log n +2$ & $7 \log n +8$\\
    \hline
    $P5$ & $\mathcal{E}(0, c_R)$& $2 \log n +6$ & $2 \log n +1$  &$2 \log n +2$ & $7 \log n +9$\\
    \hline
    $P6$ & $\mathcal{E}(0, c_B)$ & $2 \log n +7$& $2 \log n +1$  &$2 \log n +2$ & $7 \log n +10$\\
    \hline
     \hline
      \end{tabular}
      }
     \caption{{\footnotesize 
     * An integer entry $k$ in the table indicates it contains a Boolean string of length $4 \log n$ with Hamming weight $k$.\\
     The standard form of encoding of element
     $(V, (r_L, c_L), (r_R,c_R),c_B, \vdash)$ by a $96\log n$ bit string that is broken into 6 parts $P_1, \dots, P_6$ of equal size and each Part is further broken into 4 Blocks $B_1, B_2, B_3$ and $B_4$. So all total there are 24 blocks each containing a $4\log n$-bit string.
     For the standard form of encoding of element $(V, (r_L, c_L), (r_R,c_R),c_B,\dashv)$ we encode $(V, (r_L, c_L), (r_R,c_R),c_B,\vdash)$ in the standard form as described in the table and then apply the $\swap$ on each block. The last column of the table indicates the Hamming weight of each Part.}
  }\label{table:standardform}
 \end{table}

\vspace{0.8em}

\noindent\textbf{Standard-form encoding of $(V, (r_{L}, c_{L}), (r_{R}, c_{R}), (c_{B}), T)$ where $T = \vdash$: } 
For the standard-form encoding, we will assume that the information of  $V, r_{L}, c_{L}, r_{R}, c_{R}, c_{B}$ are stored in parts $P1, P2, P3, P4, P5$, and $P6$ respectively. 
For all $i \in [6]$, the part $P_i$ with have blocks $B_{1}, B_{2}, B_3$ and $B_4$, of which the block $B_1$ will be the encoding-block. The encoding will ensure that every parts within a cell will have a distinct Hamming weight. The description is also compiled in the Table~\ref{table:standardform}.

\begin{itemize}
    \item For part $P1$ (that is the encoding of $V$) the encoding block $B_1$ will store  $\ell_1 \cdot \ell_2$ where $\ell_1$ be the $2\log n$ bit binary string with Hamming weight $2\log n$ and $\ell_2$ is any $2\log n$ bit binary string with Hamming weight $2 \log n-1 - V$. The blocks $B_2$, $B_3$ and $B_4$ will store a $4\log n$ bit string that has Hamming weight $4 \log n, 2 \log n+1$ and $2\log n+2$ respectively. Any fixed string with the correct Hamming weight will do. We are not fixing any particular string for the blocks $B_2$, $B_3$, and $B_4$ to emphasize the fact that we will be only interested in the Hamming weights of these strings.
    
    \item The encoding block $B1$ for parts $P2, P3, P4, P5$ and $P6$ will store the string $\mathcal{E}(r_L, 0)$, $\mathcal{E}(0, c_L)$, $\mathcal{E}(r_R, 0)$, $\mathcal{E}(0, c_r)$ and $\mathcal{E}(0, C_B)$ respectively, where $\mathcal{E}$ is the Balanced-pointer-encoding function (Definition~\ref{def:bpencoding}).
    For part $P_i$ (with $2\leq i \leq 6$) block $B_2, B_3$ and $B_4$ will store any $4\log n$ bit string with Hamming weight $2\log n +1 + i$, $ 2\log n +1$ and $ 2 \log n +2$ respectively.
\end{itemize}

\vspace{0.8em}

\noindent\textbf{Standard form encoding of $(V, (r_{L}, c_{L}), (r_{R}, c_{R}), (c_{B}), T)$ where $T = \dashv$: } 
For obtaining a standard-form encoding of $(V, (r_{L}, c_{L}), (r_{R}, c_{R}), (c_{B}), T)$ where $T = \dashv$, first we encode $(V, (r_{L}, c_{L}), (r_{R}, c_{R}), (c_{B}), T)$ where $T = \vdash$ using the standard-form encoding.  Let $(P1, P2, \dots, P6)$ be the standard-form encoding of $(V, (r_{L}, c_{L}), (r_{R}, c_{R}), (c_{B}), T)$ where $T = \vdash$. Now for each of the blocks apply the $\swap$ operator. 

\vspace{0.8em}

\noindent\textbf{Valid permutation of the standard form:}
Now we will give a set of valid permutations to the bits of the encoding of any element of $\bar\Sigma$. The set of valid permutations is classified into 3 categories: 
\begin{enumerate}
    \item Part-permutation: The 6 parts can be permuted using any permutation from $\sym_6$
    \item Block-permutation: In each of the parts, the 4 blocks (say $B_1, B_2, B_3, B_4$) can be permuted in two ways. $(B_1, B_2, B_3, B_4)$ can be send to one of the following
    $$\mbox{(a) Simple Block Swap: } (B_3, B_4, B_1, B_2) \hspace{1.3cm} \mbox{(b) Block Flip: } (B_2, B_1, \flip(B_3), \flip(B_4))$$

\end{enumerate}

\vspace{0.3em}

\noindent\textbf{The ``decoding" function $\De: \zone^{96 \log n} \to \bar\Sigma$: } 


\begin{itemize}
    \item Identify the parts containing the encoding of $V$, $r_L$, $c_L$, $r_R$, $c_R$ and $c_B$. This is possible because every part has a unique Hamming weight.
    
    \item For each part identify the blocks. This is also possible as in any part all the blocks have distinct Hamming weight.  Recall, the valid Block-permutations, namely Simple Block Swap and Block Flip. By seeing the positions of the blocks one can 
    understand if $\flip$ was applied and to what and using that one can 
    revert the blocks back to the standard-form (recall Definition~\ref{defi: bb(i)}).
 
    \item In the part containing the encoding of $V$ consider the encoding block. 
     If the block is of the form $\{(\ell_1\ell_2)\, \text{such that}\, |\ell_1| = 2 \log n, |\ell_2| \leq 2 \log n-1\}$ then $T= \{\vdash\}$.  If the block is of the form $\{(\ell_2\ell_1)\, \text{such that}\, |\ell_1| = 2 \log n, |\ell_2| \leq 2 \log n-1\}$ then $T =\{\dashv\}$. 
     
     \item By seeing the encoding block we can decipher the original values and the pointers. 
   
     \item If the $96 \log n$ bit string doesn't have the form of a valid encoding, then decode it as $(0, \perp, \perp, \perp)$.
     
\end{itemize}

\subsection{Proof of transitivity of the function}
We start with describing the transitive group for which $F_{\ref{thm:R_0vsD}}$ is transitive. 

\vspace{0.8em}

\noindent\textbf{The Transitive Group:} We start with describing a transitive group $\mathcal{T}$ acting on the cells of the matrix $A$. The matrix has rows $r_1, \dots, r_n$ and columns $c_1, \dots, c_n$. And we use the encoding function $\mathcal{E}$ to encode the rows and columns. So the index of the rows and columns are encoded using a $4\log n$ bit string. A permutation from $\bt_{4\log n}$ (see Definition~\ref{def:transitivegroup}) on the indices of a $4\log n$ bit string will therefore induce a permutation on the set of rows and columns which will give us a permutation on the cells of the matrix. We will now describe the group $\mathcal{T}$ acting on the cells of the matrix by describing the permutation group $\hat{\mathcal{T}}$ acting on the indices of a $4\log n$ bit string. The group $\hat{\mathcal{T}}$ will be the group $\bt_{4\log n}$ acting on the set $[4\log n]$. We will assume that $\log n$ is a power of $2$. The group $\mathcal{T}$ will be the resulting group of permutations on the cells of the matrix induced by the group $\hat{\mathcal{T}}$ acting on the indices on the balanced-pointer-encoding. Note that $\mathcal{T}$ is acting on the domain of $\mathcal{E}$ and $\hat{\mathcal{T}}$ is acting on the image of $\mathcal{E}$.
Also $\hat{\mathcal{T}}$ is a transitive subgroup of $\sym_{4 \log n}$ from Claim~\ref{obs: btk is a transitive group}.

\begin{observation}\label{obs:bitflip}
For any $1\leq i\leq 2\log n$ consider the permutation  ``$i$th-bit-flip'' in  $\hat{\mathcal{T}}$ that applies the transposition $(2i-1, 2i)$ to the indices of the balanced-pointer-encoding. Since the $\mathcal{E}$-encoding of the row $(r_k,0)$ uses the balanced binary representation of $k$ in the first half and all zero string in the second half, the $j$th bit in the binary representation of $k$ is stored in the $2j-1$ and $2j$-th bit in the $\mathcal{E}$-encoding of $r_i$. So the $j$-th-bit-flip acts on the sets of rows by swapping all the rows with $1$ in the $j$-th bit of their index with the corresponding rows with $0$ in the $j$-th bit of their index. Also, if $i> \log n$ then there is no effect of the $i$-th-bit-flip operation on the set of rows. 
Similarly the $\mathcal{E}$-encoding of the column $(0,C_j)$ uses the balanced binary representation of $j$ in the second half and all zero sting in the first half. 
\end{observation}

Using Observation~\ref{obs:bitflip} we have the following claim.

\begin{claim}\label{cl:trasitiveM}
The group $\mathcal{T}$ acting on the cells of the matrix is a transitive group. That is, for all $1 \leq i_1,j_1,i_2, j_2 \leq n$ there is a permutation $\hat{\sigma}\in \hat{\mathcal{T}}$ such that $\hat{\sigma}[\mathcal{E}(i_1, 0)] = \mathcal{E}(i_2,0)$ and $\hat{\sigma}[\mathcal{E}(0,j_1)]= (0,j_2)$. In other words, there is a $\sigma\in \mathcal{T}$ acting on the cell of the matrix that would take the cell corresponding to row $r_{i_1}$ and column $c_{j_1}$ to the cell corresponding to row $r_{i_2}$ and column $c_{j_2}$.
\end{claim}

From the Claim~\ref{cl:trasitiveM} we see the group $\mathcal{T}$ acting on the cells of of the matrix is a transitive, but it does not touch the contents within the cells of the matrix. The input to the function $F_{\ref{thm:R_0vsD}}$ contains element of $\Gamma = \{0,1\}^{96\log n}$ in each cell. So we now need to extend the group $\mathcal{T}$ to a group $\Gr$ that acts on all the indices of all the bits of the input to the function $F_{\ref{thm:R_0vsD}}$.

Recall that the input to the function $F_{\ref{thm:R_0vsD}}$ is a $(n\times n)$-matrix with each cell of matrix containing a binary string of length $96\log n$ which has 6 parts of size $16\log n$ each and each part has 4 blocks of size $4\log n$ each. We classify the generating elements of the group $\Gr$ into 4 categories: 
\begin{enumerate}
    \item Part-permutation: In each of the cells the 6 parts can be permuted using any permutation from $\sym_6$
    \item Block-permutation: In each of the Parts the 4 blocks can be permuted in the following ways. $(B_1, B_2, B_3, B_4)$ can be send to one of the following
    \begin{enumerate}
        \item Simple Block Swap: $(B_3, B_4, B_1, B_2)$
        \item Block Flip ($\# 1$): $(B_2, B_1, \flip(B_3), \flip(B_4))$
        \item Block Flip ($\# 2$)\footnote{Actually this Block flip can be generated by a combination of Simple Block Swap and Block Flip ($\# 1$)}: $(\flip(B_1), \flip(B_2), B_4, B_3)$
    \end{enumerate}
    \item Cell-permutation: for any $\sigma \in \mathcal{T}$  the following two action has to be done simultaneously:
    \begin{enumerate}
        \item (Matrix-update) Permute the cells in the matrix according to the permutation $\sigma$. This keeps the contents within each cell untouched - it just changes the location of the cells. 
        \item (Pointer-update) For each of the blocks in each of the parts in each of the cells permute the indices of the $4\log n$-bit strings according to $\sigma$, that applies $\hat{\sigma} \in \hat{\mathcal{T}}$ corresponding to $\sigma$.
    \end{enumerate}
\end{enumerate}

We now have the following theorems that would prove that the function $F_{\ref{thm:R_0vsD}}$ is transitive.

\begin{theorem}\label{thm:transitivegroup}
$\Gr$ is a transitive group and the
function $F_{\ref{thm:R_0vsD}}$ is invariant under the action of the $\Gr$.
\end{theorem}

\begin{proof}[of Theorem~\ref{thm:transitivegroup}]

To prove that the group $\Gr$ is transitive we show that for any indices $p, q \in [96n^2\log n]$ there is a permutation $\sigma\in \Gr$ that would take $p$ to $q$. Recall that the string $\{0,1\}^{96n^2\log n}$ is a matrix $\Gamma^{(n\times n)}$ with $\Gamma = \{0,1\}^{96\log n}$ and every element in $\Gamma$ is broken into 6 parts and each part is broken into 4 blocks of size $4\log n$ each. So we can think of the index $p$ as sitting in $k_p$th position ($1\leq k_p\leq 4\log n$) in the block $B_p$ of the part $P_p$ in the $(r_p, c_p)$-th cell of the matrix. Similarly, we can think of  $q$ as sitting in $k_q$th position ($1\leq k_q\leq 4\log n$) in the block $B_q$ of the part $P_q$ in the $(r_q, c_q)$-th cell of the matrix. 

We will give a step-by-step technique in which permutations from $\Gr$ can be applied to move $p$ to $q$.

\begin{enumerate}
    \item[Step 1] \textbf{Get the positions in the block correct:} If $k_p \neq k_q$ then take a permutation $\hat{\sigma}$ from $\hat{\mathcal{T}}$ that takes $k_p$ to $k_q$.
    Since $\hat{\mathcal{T}}$ is a transitive group, such a permutation exists. Apply the cell-permutation $\sigma \in \mathcal{T}$ corresponding to $\hat{\sigma}$. As a result, the index $p$ can be moved to a different cell in the matrix but, by the choice of $\hat{\sigma}$ its position in the block in which it is will be $k_q$. 
    Without loss of generality, we assume the cell location does not change.
    
    \item[Step 2] \textbf{Get the cell correct:} Using a cell-permutation that corresponds to a series of ``bit-flip'' operations change $r_p$ to $r_q$ and $c_p$ to $c_q$. Since one-bit-flip operations change one bit in the binary representation of the index of the row or column such a series of operations can be made. 
    
    Since each bit-flip operation is executed by applying the bit-flips in each of the blocks this might have once again changed the position of the index $p$ in the block. Note that, even if the position in the block changes it must be a $\flip$ operation away. In other words, since at the beginning of this step $k_p = k_q$, so if $k_q$ is even (or odd) then after the series bit-flip operations the position of $p$ in the block is either  $k_q$ or  $(k_q -1)$ (or $(k_q + 1)$). 
    
    \item[Step 3] \textbf{Align the Part:} Apply a suitable permutation to ensure that the part $P_p$ moves to part $P_q$. Note this does not change the cell or the block within the part of the position in the block.
    
    \item[Step 4] \textbf{Align the Block:} Using a suitable combination of Simple Block Swap and Block Flip ensures the Block number gets matched, that is $B_p$ goes to $B_q$. In this case, the cell or the Part does not change. Depending on whether the Block Flip operation is applied the position in the block can again change. Note that, the current position in the block  $k_p$ is at most one $\flip$ away from $k_q$.
    
    \item[Step 5] \textbf{Apply the final flip:} It might so happen that already we a done after the last step. If not we know that the current position in the block  $k_p$ is at most one $\flip$ away from $k_q$. So we apply the suitable Block-flip operation. This will not change the cell position, Part Number, or Block number and the position in the block will match. 
\end{enumerate}
Hence we have proved that the group $\Gr$ is transitive. Now we show that the function 
 $F_{\ref{thm:R_0vsD}}$ is invariant under the action of $\Gr$, i.e., for any elementary operations $\pi$ from the group $\Gr$ and for any input $\Gamma^{(n\times n)}$ the function value does not change even if after the input is acted upon by the permutation $\pi$. 

\textbf{Case 1: $\pi$ is a Part-permutation:} It is easy to see that the decoding algorithm $\De$ is invariant under Part-permutation. This was observed in the description of the decoding algorithm $\De$ in Section~\ref{sec: encoding scheme of variant 1}. So clearly the function $F_{\ref{thm:R_0vsD}}$ is invariant under any Part-permutation.

\textbf{Case 2: $\pi$ is a Block-permutation:} Here also it is easy to see that the decoding algorithm $\De$ is invariant under Block-permutation. This was observed in the description of the decoding algorithm $\De$ in Section~\ref{sec: encoding scheme of variant 1}.  Thus $F_{\ref{thm:R_0vsD}}$ is also invariant under any block permutation.

\textbf{Case 3: $\pi$ is a Cell-permutation} From Observation~\ref{thm: invariance of variant 1} it is enough to prove that when we permute the cells of the matrix we update the points in the cells accordingly. 

Let $\pi \in \mathcal{T}$ be a permutation that permutes only the rows of the matrix. By Claim~\ref{cl:trick}, we see that the contents of the cells will be updated accordingly. Similarly if $\pi$ only permute the columns of the matrix we will be fine.

Finally, if $\pi$ swaps the row set and the column set (that is if $\pi$ makes a transpose of the matrix) then for all $i$ row $i$ is swapped with column $i$ and it is easy to see that $\hat{\pi}[\mathcal{E}(i,0)] =\mathcal{ E}(0,i)$. In that case, the encoding block of the value part in a cell also gets swapped. This will thus be encoding the $T$ value as $\dashv$. And so the function value will not be affected as the $T = \dashv$ will ensure that one should apply the 
$\pi$ that swaps the row set and the column set to the input before evaluating the function. 
\end{proof}

\subsection{Properties of the function}

\begin{claim}
\label{claim: D(Fnn)}
Deterministic query complexity of $F_{\ref{thm:R_0vsD}}$ is ${\Omega}(n^2)$.
\end{claim}

\begin{proof}
The function $\mabbl_{(n, n)}$ is a ``harder'' function than $\abbl_{(n,n)}$. Here by ``harder'' we mean $\abbl_{(n,n)}$ is a sub-function of  $\mabbl_{(n, n)}$. $\mabbl_{(n, n)}$ contains extra alphabets compared to $\abbl_{(n,n)}$, upon restriction on those extra alphabets we can get back the function $\abbl_{(n,n)}$, so $\dqc(\mabbl_{(n,n)})$ is at least that of $\dqc(\abbl_{(n, n)})$. Now since, $F_{\ref{thm:R_0vsD}}$ is $\left(\mabbl_{(n, n)}\circ \De\right)$ so clearly the $\dqc(F_{\ref{thm:R_0vsD}})$ is at least $\dqc(\abbl_{(n, n)})$.  Theorem~\ref{thm:D(abbl)} proves that $\dqc(\abbl_{(n,n)})$ is $\Omega(n^2)$. Hence $\dqc(F_{\ref{thm:R_0vsD}}) = \Omega(n^2)$.
\end{proof}

The following Claim~\ref{cl:mabbl} follows from the definition of the function  $\mabbl_{(n,n)}$.
 \begin{claim}\label{cl:mabbl}
The following are some properties of the function $\mabbl_{(n,n)}$
\begin{enumerate}
 \item $\roc(\mabbl_{(n, n)}) \leq 3\roc(\abbl_{(n, n)}) $
 \item $\qqc(\mabbl_{(n, n)}) \leq 3\qqc(\abbl_{(n, n)}) $
 \item $\deg(\mabbl_{(n, n)}) \leq 3\deg(\abbl_{(n, n)}) $
\end{enumerate}
\end{claim}

Finally, from Theorem~\ref{theo: composition} we see that the $\roc(F_{\ref{thm:R_0vsD}})$, $\qqc(F_{\ref{thm:R_0vsD}})$ and $\deg(F_{\ref{thm:R_0vsD}})$ are at most $O(\roc(\mabbl_{(n,n)}\cdot \log n)$, $O(\qqc(\mabbl_{(n, n)}\cdot \log n)$ and $O(\deg(\mabbl_{(n, n)}\cdot \log n)$, respectively. So combining this fact with Claim~\ref{claim: D(Fnn)},  Claim~\ref{cl:mabbl} and Theorem~\ref{thm: Properties of variant 1} (from \cite{ABBL+17}) we have Theorem~\ref{thm:R_0vsD}.

\section{Separation between sensitivity and randomized query complexity }
\label{sec: s vs qqc}
\cite{DHT17} showed that functions that witness a gap between deterministic query complexity (or randomized query complexity), and $\ucmin$ can be transformed to give functions that witness separation between deterministic query complexity (or randomized query complexity) and sensitivity. We observe that transformation the~\cite{DHT17} described preserves transitivity.
Our transitive functions from Theorem~\ref{thm:R_0vsD} along with the transformation from~\cite{DHT17} gives the cubic separations between $\rqc$ and $\mathsf{s}$


We start by defining \emph{desensitization transform} of Boolean functions as defined in~\cite{DHT17}.

\begin{defi}[Desensitized Transformation]
\label{defi:desensitization}
Let $f :\zone^n \to \zone$. Let $U$ be a collection of unambiguous $1$-certificates for $f$, each of size at most $\mathsf{UC}_1(f )$. For each $x \in f^{-1}(1) $, let $p_x \in U$
be the unique certificate in $U$ consistent with $x$. The desensitized version of $f$ is the function $f_{\desentransform}: \zone^{3n} \to \zone$ defined by $f_{\desentransform}(x_1 x_2 x_3) = 1$ if and only if $f(x_1) = f(x_2) = f(x_3) = 1$ and $p_{x_1} = p_{x_2} = p_{x_3}$ .
\end{defi}

\begin{observation}
\label{obs: desensitized version is transitive}
If $f:\zone^n \to \zone$ is transitive, then $f_{\desentransform}:\zone^{3n} \to \zone$ is also transitive.
\end{observation}
\begin{proof}
Let $T_f \subseteq \sym_n$ be the transitive group corresponding to $f$ and let $x_1 x_2 x_3 \in \zone^{3n}$ be the input to $f_{\desentransform}$. Consider the following permutations acting on the input $x_1 x_2 x_3$ to $f_{\desentransform}$:
\begin{enumerate}
    \item $\sym_3$ acting on the indices $\{1,2,3\}$ and
    \item $\{(\sigma, \sigma, \sigma) \in \sym_{3n} | \sigma \in T_f\}$ acting on $(x_1, x_2, x_3)$.
\end{enumerate}
Observe that the above permutations act transitively on the inputs to $f_{\desentransform}$. 
Also from the definition of $f_{\desentransform}$ the value of the function $f_{\desentransform}$ is invariant under these permutations.
\end{proof}

Next, we need the following theorem from~\cite{DHT17}. The theorem is true for more general complexity measures. We refer the reader to~\cite{DHT17} for a more general statement.

\begin{theorem}[\cite{DHT17}]
\label{thm: Tal's result}
For any $k \in \mathbb{R}^+$, if there is a family of function with $\dqc(f) = \widetilde{\Omega}(\mathsf{UC}_{min}(f)^{1+ k})$, then for the family of functions defined by $\widetilde{f} = \OR_{3 \UC_{\min}(f)} \circ f_{\desentransform}$ satisfies $\dqc(\widetilde{f}) = \widetilde{\Omega}(s(\widetilde{f})^{2+k}) $. Also, if we replace $\dqc(f)$ by $\rqc(f)$, $\qqc(f)$ or $\cert(f)$, we will get the same result.
\end{theorem}
\begin{proof}[of Theorem~\ref{theo: intro sensitivity separations}]
Let us begin with the transitive functions $F_{\ref{thm:R_0vsD}}$ from Section~\ref{subsec: F1.1} which will desensitize to get the desired claim. 
From Theorem~\ref{thm:D(abbl)} and Observation~\ref{obs: UC(Fnn)} we have $\dqc(F_{\ref{thm:R_0vsD}}) \geq \widetilde{\Omega}(\UC_{min}(F_{\ref{thm:R_0vsD}})^{2})$. Note that for our function $\UC_{min} = \UC_1$. 

Let $F_{\desentransform}$ be the desensitized version of  
$F_{\ref{thm:R_0vsD}}$. Define $F_{\ref{theo: intro sensitivity separations}}$ to be $\OR_{3 \UC_{\min}(F_{\desentransform})} \circ F_{\desentransform}$.
From Theorem~\ref{thm: Tal's result}, Observation~\ref{obs: desensitized version is transitive} we have the theorem $F_{\ref{theo: intro sensitivity separations}}$ being the required function.
\end{proof}

\section{Separation between quantum query complexity and certificate complexity}
\label{sec: qcc v. certificate}
\cite{ABK16} constructed functions that demonstrated quadratic separation between quantum query complexity and certificate complexity. Their function was not transitive. We modify their function to obtain a transitive function that gives a similar separation.



We start this section by constructing an encoding scheme for the inputs to the $\ksum$ function such that the resulting function $\encodedksum$ is transitive. We then, similar to~\cite{ABK16}, define $\mathsf{ENC}-\blockksum$ function. Composing $\encodedksum$ with $\mathsf{ENC}-\blockksum$ as outer function gives us $F_{\ref{thm:Qvs C}}$.

\subsection{Function definition}
Recall that, from Definition~\ref{defi: k-sum}, for $\Sigma= [n^k]$, the function $\ksum: \Sigma^n \to \zone$ is defined as follows: on input $x_1,x_2, \dots, x_n \in \Sigma$, if there exists $k$ element $x_{i_1}, \dots, x_{i_k}$, $i_1, \dots, i_k \in [n]$, that sums to $0~(\text{mod}~|\Sigma|)$ then output $1$, otherwise output $0$. We first define an encoding scheme for $\Sigma$.

\subsubsection{Encoding scheme}
\label{sec: Encoding3}
Similar to Section~\ref{sec: encoding scheme of variant 1} we first define the standard form of the encoding of $x \in \Sigma$ and then extend it by action of suitable group action to define all encodings that represent $x \in \Sigma$ where $\Sigma$ is of size $n^k$ for some $k \in \N$.

Fix some $x \in \Sigma$ and let $x = x_1 x_2 \dots x_{k \log n}$ be the binary representation of $x$. The standard form of encoding of $x$ is defined as follows: For all $i\in [ k\log n]$ we encode $x_i$ with with $4(k \log n +2)$ bit Boolean string satisfying the following three conditions:
\begin{enumerate}
    \item $x_i= x_{i1} x_{i2} x_{i3} x_{i4}$ where each $x_{ij}$, for $j \in [4]$, is a $(k\log n+2)$ bit string,
    
    \item  if $x_i =1$ then $|x_{i1}|=1, |x_{i2}|=0 , |x_{i3}|=2, |x_{i4}|= i + 2$, and
    
    \item   if $x_i =0$ then $|x_{i1}|=0, |x_{i2}|=1, |x_{i3}|=2, |x_{i4}|= i + 2$.
\end{enumerate}
Having defined the standard form, other valid encodings of $x_i= (x_{i1} x_{i2} x_{i3} x_{i4})$ are obtained by the  action of permutations $(12)(34), (13)(24) \in \sym_4$ on the indices $\{i1, i2,i3, i4\}$.
Finally if $x = \{(x_{ij}| i\in [ k\log n], j \in [4]\} $, then $\{(x_{\sigma(i)\gamma(j)})| \sigma \in \sym_{k\log n}, \gamma \in T \subset \sym_4\}$ is the set of all valid encoding for $x \in \Sigma$.

The decoding scheme follows directly from the encoding scheme. Given $y \in \zone^{k \log n(4k \log n + 8)}$, first break $y$ into $k \log n $ blocks each of size $4k \log n + 8$ bits. If each block is a valid encoding then output the decoded string else output that $y$ is not a valid encoding for any element from $\Sigma$.

\subsubsection{Definition of the encoded function}
$\encodedksum$ is a Boolean function that defined on $n$-bit as follows: Split the $n$-bit input into block of size $4k\log n(k \log n +2)$. We say such a block is a valid block iff it follows the encoding scheme in Section~\ref{sec: Encoding3} i.e. represents a number from the alphabet $\Sigma$. The output value of the function is $1$ iff there exists $k$ valid block such that the number represented by the block in $\Sigma$ sums to $0~(\text{mod}~ |\Sigma|)$.

$\mathsf{ENC}-\blockksum$ is a special case of the $\encodedksum$  function. We define it next. 
$\mathsf{ENC}-\blockksum$ is a Boolean function that is defined on $n$-bit as follows: The input string is split into blocks of size $4k\log n(k \log n +2)$ and we say such block is a valid block iff it follows the encoding scheme of Section~\ref{sec: Encoding3} i.e. represents a number from the alphabet $\Sigma$. The output value of the function is $1$ iff there exists $k$ valid block such that the number represented by the block in $\Sigma$ sums to $0(\text{mod}~ |\Sigma|)$ and the number of $1$ in the other block is at least $6 \times (k \log n)$. Finally, similar to~\cite{ABK16} define $F_{\ref{thm:Qvs C}}:\zone^{n^2} \to \zone$ to be $\mathsf{ENC}-\blockksum \circ \encodedksum$, with $k = \log n$.

The proof of Theorem~\ref{thm:Qvs C} is the same as that of~\cite{ABK16}. We give the proof here for completeness.

\begin{proof}[of Theorem~\ref{thm:Qvs C}]
We first show that the certificate complexity of $F_{\ref{thm:Qvs C}}$ is $\mathsf{O}(4 n k^2 \log n(k \log n+2)$. For this we show that every input to $\mathsf{ENC}-\blockksum$, the outer function of $F_{\ref{thm:Qvs C}}$, has a certificate with $\widetilde{O}(k \times (4k\log n(k \log n +2)))$ many $0$'s and $O(n)$ many $1$'s. Also the inner function of $F_{\ref{thm:Qvs C}}$, i.e. $\encodedksum$, has $1$-certificate of size $\mathsf{O}(4k^2 \log n(k \log n+2)$ and $0$-certificate of size $\mathsf{O}(n)$. Hence, the $F_{\ref{thm:Qvs C}}$ function has certificate of size $\mathsf{O}(4 n k^2 \log n(k \log n+2))$.

Every $1$-input of $\mathsf{ENC}-\blockksum$ has $k$ valid encoded blocks such that the number represented by them sums to $(0~\text{mod}~|\Sigma|)$. This can be certified using at most $\widetilde{\mathsf{O}}(k \times (4k\log n(k \log n +2)))$ number of $0$'s and all the $1$'s from every other block. 

There are two types of $0$-inputs of $\mathsf{ENC}-\blockksum$. The first type of $0$-input has at least one block in which a number of $1$ is less than $6 \times (k \log n)$ and the zeros of that block is a $0$-certificate of size $\widetilde{O} (4k\log n(k \log n +2))$.
The other type of $0$-input is such that every block contains at least $6 \times (k \log n)$ number of $1$'s. This type of $0$-input can be certified by providing all the $1$'s in every block, which is at most $\mathsf{O}(n)$. This is because using all the $1$'s we can certify that even if the blocks were valid, no $k$-blocks of them are such that the number represented by them sums to $0~(\text{mod}~|\Sigma|)$.  

Next, we prove $\Omega(n^2)$ lower bound on the quantum query complexity of $F_{\ref{thm:Qvs C}}$. From Theorem~\ref{theo: tight composition}, $\qqc(F_{\ref{thm:Qvs C}}) = \Omega(\qqc(\mathsf{ENC}-\blockksum) \qqc(\encodedksum))$. Note that the inputs to  $\ksum$ are coming from some alphabet set $\Sigma$ where size of $\Sigma$ is at most $n^k$. To represent any such $x \in \sigma$ we need at most $k \log n$ many bits. Now if we take the alphabet set a bit larger that is if we need $10 log n$ bits to represent any element from the larger alphabet set, then $\ksum$ a sub-function of $\mathsf{\blockksum}$. Now our encoding scheme being again of size $O(k\log n)$ it follows that $\encodedksum$ is a sub-function of $\mathsf{ENC-\blockksum}$. Since quantum query complexity of $\mathsf{ENC-\blockksum}$ is ${\Omega}(\qqc(\encodedksum/ {k \log n}))$ , from Theorem~\ref{theo: ksum qcc lower bound} the quantum query complexity of the $\mathsf{ENC-\blockksum}$ function is $\Omega\left(\frac{n^{\frac{k}{k+1}}}{k^{\frac{3}{2}}\log n(k \log n +2)}\right)$. Thus
\begin{align*}
    \qqc(F_{\ref{thm:Qvs C}}) = \Omega\left(\frac{n^{\frac{2k}{k+1}}}{k^{3}\log n(k \log n +2)}\right).
\end{align*}
Hence, $\qqc(F_{\ref{thm:Qvs C}})= \widetilde{\Omega}(n^2)$, taking $k= \log n$.
\end{proof}


\section{Conclusion}
As far as we know, this is the first paper that presents a thorough investigation on the relationships between various pairs of complexity measures for transitive function. 

The current best-known relationships and best-known separations between various pairs of measures for transitive functions are summarized in Table~\ref{table: main_table}. 
Unfortunately, a number of cells in the table is not tight. 
In this context, we would like to point out some important directions:
\begin{itemize}
    \item For some of these cells, the separation results for transitive functions are weaker than that of the general functions. A natural question is the following: why can't we design a transitive version of the general functions that achieve the same separation?
    Thus following is a natural question.
    \begin{open}
    For a pair of complexity measures for Boolean functions whose best-known separations are achieved via cheat sheets, obtain similar separations for transitive Boolean functions.
    \end{open}

    \item A total function was constructed in \cite{BT} that demonstrates quadratic separations between approximate degrees with sensitivity and several other complexity measures. It is thus natural to investigate the following open problem.
    \begin{open}
    Come up with transitive functions that achieve similar separations for those pair of measures whose best-known separations are shown by~\cite{BT}.
    \end{open}
  
    \item Recently \cite{BGJ+21}, \cite{Balodis21} and \cite{BDG+22} came up with new classes of Boolean functions, starting with the HEX (see~\cite{BGJ+21}) and EAH (see~\cite{BDG+22}) functions,
    that exhibit improved separations between certificate complexity and other complexity measures using the. 
    
    In light of these recent developments is important to ask whether similar separations can be shown for transitive functions. Following an open problem is a natural starting point.
    \begin{open}
    Can the HEX and EAH functions be modified to transitive functions, while preserving their desired complexity measures up to poly-logarithmic factors?
    \end{open}
    
    \item While we have been concerned only with lower bounds in this paper, it is an exciting research direction to bridge the gap between complexity measures of transitive Boolean functions by providing improved upper bounds.
    \begin{open}
    Bridge the gaps in Table~\ref{table: main_table} by coming up with better upper bounds on complexity measures for transitive functions.
    \end{open}
    \end{itemize}
  Even with the recent results of~\cite{Huang} and~\cite{ABK+}, there are significant gaps between the best-known lower and upper bounds in this case which gives another set of open problems to investigate in the study of combinatorial measures of transitive Boolean functions.
   

\nocite{*}
\bibliographystyle{abbrvnat}
\bibliography{ref}

\begin{thebibliography}{46}
\providecommand{\natexlab}[1]{#1}
\providecommand{\url}[1]{\texttt{#1}}
\expandafter\ifx\csname urlstyle\endcsname\relax
  \providecommand{\doi}[1]{doi: #1}\else
  \providecommand{\doi}{doi: \begingroup \urlstyle{rm}\Url}\fi

\bibitem[Aaronson(2008)]{Aar08}
S.~Aaronson.
\newblock Quantum certificate complexity.
\newblock \emph{Journal of Computer and System Sciences}, 74\penalty0 (3):\penalty0 313--322, 2008.
\newblock \doi{10.1016/j.jcss.2007.06.020}.

\bibitem[Aaronson et~al.(2016)Aaronson, Ben{-}David, and Kothari]{ABK16}
S.~Aaronson, S.~Ben{-}David, and R.~Kothari.
\newblock Separations in query complexity using cheat sheets.
\newblock In \emph{{STOC}}, pages 863--876, 2016.
\newblock \doi{10.1145/2897518.2897644}.

\bibitem[Aaronson et~al.(2020)Aaronson, Ben{-}David, Kothari, and Tal]{ADKT}
S.~Aaronson, S.~Ben{-}David, R.~Kothari, and A.~Tal.
\newblock Quantum implications of {H}uang's sensitivity theorem.
\newblock \emph{CoRR}, abs/2004.13231, 2020.
\newblock URL \url{https://arxiv.org/abs/2004.13231}.

\bibitem[Aaronson et~al.(2021)Aaronson, Ben{-}David, Kothari, Rao, and Tal]{ABK+}
S.~Aaronson, S.~Ben{-}David, R.~Kothari, S.~Rao, and A.~Tal.
\newblock Degree vs. approximate degree and quantum implications of {H}uang's sensitivity theorem.
\newblock In \emph{{STOC}}, pages 1330--1342, 2021.
\newblock \doi{10.1145/3406325.3451047}.

\bibitem[Ambainis(2016)]{Ambainis16}
A.~Ambainis.
\newblock Superlinear advantage for exact quantum algorithms.
\newblock \emph{{SIAM} Journal on Computing}, pages 617--631, 2016.
\newblock \doi{10.1137/130939043}.

\bibitem[Ambainis et~al.(2017)Ambainis, Balodis, Belovs, Lee, Santha, and Smotrovs]{ABBL+17}
A.~Ambainis, K.~Balodis, A.~Belovs, T.~Lee, M.~Santha, and J.~Smotrovs.
\newblock Separations in query complexity based on pointer functions.
\newblock \emph{Journal of the {ACM}}, 64\penalty0 (5):\penalty0 32:1--32:24, 2017.
\newblock \doi{10.1145/3106234}.

\bibitem[Arora and Barak(2009)]{CCBB}
S.~Arora and B.~Barak.
\newblock \emph{Computational Complexity - {A} Modern Approach}.
\newblock Cambridge University Press, 2009.
\newblock ISBN 978-0-521-42426-4.
\newblock URL \url{http://www.cambridge.org/catalogue/catalogue.asp?isbn=9780521424264}.

\bibitem[Balodis(2021)]{Balodis21}
K.~Balodis.
\newblock Several separations based on a partial {Boolean} function.
\newblock \emph{arXiv:2103.05593}, 2021.
\newblock URL \url{https://arxiv.org/abs/2103.05593}.

\bibitem[Balodis et~al.(2021)Balodis, Ben-David, Göös, Jain, and Kothari]{BDG+22}
K.~Balodis, S.~Ben-David, M.~Göös, S.~Jain, and R.~Kothari.
\newblock Unambiguous {DNF}s and {Alon-Saks-Seymour}.
\newblock In \emph{FOCS}, pages 116--124, 2021.
\newblock \doi{10.1109/FOCS52979.2021.00020}.

\bibitem[Beals et~al.(2001)Beals, Buhrman, Cleve, Mosca, and de~Wolf]{BBCM+01}
R.~Beals, H.~Buhrman, R.~Cleve, M.~Mosca, and R.~de~Wolf.
\newblock Quantum lower bounds by polynomials.
\newblock \emph{J. {ACM}}, 48\penalty0 (4):\penalty0 778--797, 2001.
\newblock URL \url{https://doi.org/10.1145/502090.502097}.

\bibitem[Belovs and Spalek(2013)]{BSK13}
A.~Belovs and R.~Spalek.
\newblock Adversary lower bound for the k-sum problem.
\newblock In \emph{{ITCS}}, pages 323--328, 2013.
\newblock \doi{10.1145/2422436.2422474}.

\bibitem[Ben{-}David et~al.(2017)Ben{-}David, Hatami, and Tal]{DHT17}
S.~Ben{-}David, P.~Hatami, and A.~Tal.
\newblock Low-sensitivity functions from unambiguous certificates.
\newblock In \emph{{ITCS}}, volume~67, pages 28:1--28:23, 2017.
\newblock \doi{10.4230/LIPIcs.ITCS.2017.28}.

\bibitem[Ben{-}David et~al.(2020)Ben{-}David, Childs, Gily{\'{e}}n, Kretschmer, Podder, and Wang]{DCK+20}
S.~Ben{-}David, A.~M. Childs, A.~Gily{\'{e}}n, W.~Kretschmer, S.~Podder, and D.~Wang.
\newblock Symmetries, graph properties, and quantum speedups.
\newblock In \emph{{FOCS}}, pages 649--660, 2020.
\newblock \doi{10.1109/FOCS46700.2020.00066}.

\bibitem[Ben-David et~al.(2021)Ben-David, G{\"o}{\"o}s, Jain, and Kothari]{BGJ+21}
S.~Ben-David, M.~G{\"o}{\"o}s, S.~Jain, and R.~Kothari.
\newblock Unambiguous {DNF}s from {Hex}.
\newblock \emph{Electronic Colloquium on Computational Complexity}, 28:\penalty0 16, 2021.

\bibitem[Bennett et~al.(1997)Bennett, Bernstein, Brassard, and Vazirani]{bennett1997strengths}
C.~H. Bennett, E.~Bernstein, G.~Brassard, and U.~Vazirani.
\newblock Strengths and weaknesses of quantum computing.
\newblock \emph{SIAM journal on Computing}, 26\penalty0 (5):\penalty0 1510--1523, 1997.
\newblock \doi{10.1137/S0097539796300933}.

\bibitem[Brassard et~al.(2002)Brassard, Hoyer, Mosca, and Tapp]{brassard2002quantum}
G.~Brassard, P.~Hoyer, M.~Mosca, and A.~Tapp.
\newblock Quantum amplitude amplification and estimation.
\newblock \emph{Contemporary Mathematics}, 305:\penalty0 53--74, 2002.

\bibitem[Buhrman and de~Wolf(2002)]{BW}
H.~Buhrman and R.~de~Wolf.
\newblock Complexity measures and decision tree complexity: a survey.
\newblock \emph{Theoretical Computer Science}, 288\penalty0 (1):\penalty0 21--43, 2002.
\newblock \doi{10.1016/S0304-3975(01)00144-X}.

\bibitem[Bun and Thaler(2020)]{BT}
M.~Bun and J.~Thaler.
\newblock A nearly optimal lower bound on the approximate degree of {AC}\({}^{\mbox{0}}\).
\newblock \emph{{SIAM} Journal on Computing}, 49\penalty0 (4), 2020.
\newblock \doi{10.1137/17M1161737}.

\bibitem[Chakraborty(2011)]{Chakraborty11}
S.~Chakraborty.
\newblock On the sensitivity of cyclically-invariant {Boolean} functions.
\newblock \emph{Discrete Mathematics and Theoretical Computer Science}, 13\penalty0 (4):\penalty0 51--60, 2011.
\newblock \doi{10.46298/dmtcs.552}.

\bibitem[Chakraborty et~al.(2021)Chakraborty, Kayal, and Paraashar]{CKP21}
S.~Chakraborty, C.~Kayal, and M.~Paraashar.
\newblock Separations between combinatorial measures for transitive functions.
\newblock \emph{arXiv:2103.12355}, 2021.
\newblock URL \url{https://arxiv.org/abs/2103.12355}.

\bibitem[Chakraborty et~al.(2022)Chakraborty, Kayal, and Paraashar]{SKP22}
S.~Chakraborty, C.~Kayal, and M.~Paraashar.
\newblock Separations between combinatorial measures for transitive functions.
\newblock In \emph{{ICALP}}, volume 229, pages 36:1--36:20, 2022.
\newblock URL \url{https://doi.org/10.4230/LIPIcs.ICALP.2022.36}.

\bibitem[Drucker(2011)]{Drucker11}
A.~Drucker.
\newblock Block sensitivity of minterm-transitive functions.
\newblock \emph{Theoretical Computer Science}, 412\penalty0 (41):\penalty0 5796--5801, 2011.
\newblock \doi{10.1016/j.tcs.2011.06.025}.

\bibitem[Gao et~al.(2013)Gao, Mao, Sun, and Zuo]{GaoMSZ13}
Y.~Gao, J.~Mao, X.~Sun, and S.~Zuo.
\newblock On the sensitivity complexity of bipartite graph properties.
\newblock \emph{Theoretical Computer Science}, 468:\penalty0 83--91, 2013.
\newblock \doi{10.1016/j.tcs.2012.11.006}.

\bibitem[Gilmer et~al.(2016)Gilmer, Saks, and Srinivasan]{GSS16}
J.~Gilmer, M.~E. Saks, and S.~Srinivasan.
\newblock Composition limits and separating examples for some {Boolean} function complexity measures.
\newblock \emph{Combinatorica}, 36\penalty0 (3):\penalty0 265--311, 2016.
\newblock \doi{10.1007/s00493-014-3189-x}.

\bibitem[G{\"{o}}{\"{o}}s et~al.(2018{\natexlab{a}})G{\"{o}}{\"{o}}s, Jayram, Pitassi, and Watson]{GPJW18}
M.~G{\"{o}}{\"{o}}s, T.~S. Jayram, T.~Pitassi, and T.~Watson.
\newblock Randomized communication versus partition number.
\newblock \emph{{ACM } Transactions on Computation Theory}, 10\penalty0 (1):\penalty0 4:1--4:20, 2018{\natexlab{a}}.
\newblock \doi{10.1145/3170711}.

\bibitem[G{\"{o}}{\"{o}}s et~al.(2018{\natexlab{b}})G{\"{o}}{\"{o}}s, Pitassi, and Watson]{GPW18}
M.~G{\"{o}}{\"{o}}s, T.~Pitassi, and T.~Watson.
\newblock Deterministic communication vs. partition number.
\newblock \emph{{SIAM} Journal on Computing}, 47\penalty0 (6):\penalty0 2435--2450, 2018{\natexlab{b}}.
\newblock \doi{10.1137/16M1059369}.

\bibitem[Huang(2019)]{Huang}
H.~Huang.
\newblock Induced subgraphs of hypercubes and a proof of the sensitivity conjecture.
\newblock \emph{Annals of Mathematics}, 190\penalty0 (3):\penalty0 949--955, 2019.
\newblock \doi{10.4007/annals.2019.190.3.6}.

\bibitem[Kimmel(2013)]{Kimmel13}
S.~Kimmel.
\newblock Quantum adversary (upper) bound.
\newblock \emph{Chicago Journal of Theoritical Computer Science}, 2013.
\newblock \doi{10.4086/cjtcs.2013.004}.

\bibitem[Kulkarni and Tal(2016)]{KT16}
R.~Kulkarni and A.~Tal.
\newblock On fractional block sensitivity.
\newblock \emph{Chicago Journal of Theoritical Computer Science}, 2016.
\newblock URL \url{http://cjtcs.cs.uchicago.edu/articles/2016/8/contents.html}.

\bibitem[Lee and Roland(2013)]{lee2013strong}
T.~Lee and J.~Roland.
\newblock A strong direct product theorem for quantum query complexity.
\newblock \emph{Computational Complexity}, 22\penalty0 (2):\penalty0 429--462, 2013.
\newblock \doi{https://doi.org/10.1007/s00037-013-0066-8}.

\bibitem[Lee et~al.(2011)Lee, Mittal, Reichardt, Spalek, and Szegedy]{LMR+11}
T.~Lee, R.~Mittal, B.~W. Reichardt, R.~Spalek, and M.~Szegedy.
\newblock Quantum query complexity of state conversion.
\newblock In \emph{{FOCS}}, pages 344--353, 2011.
\newblock \doi{10.1109/FOCS.2011.75}.

\bibitem[Li and Sun(2017)]{LiS17}
Q.~Li and X.~Sun.
\newblock On the sensitivity complexity of k-uniform hypergraph properties.
\newblock In \emph{{STACS}}, volume~66, pages 51:1--51:12, 2017.
\newblock \doi{10.4230/LIPIcs.STACS.2017.51}.

\bibitem[Montanaro(2014)]{Montanaro14}
A.~Montanaro.
\newblock A composition theorem for decision tree complexity.
\newblock \emph{Chicago Journal of Theoritical Computer Science}, 2014.
\newblock \doi{10.4086/cjtcs.2014.006}.

\bibitem[Mukhopadhyay and Sanyal(2015)]{MS}
S.~Mukhopadhyay and S.~Sanyal.
\newblock Towards better separation between deterministic and randomized query complexity.
\newblock In \emph{{FSTTCS}}, volume~45, pages 206--220, 2015.
\newblock \doi{10.4230/LIPIcs.FSTTCS.2015.206}.

\bibitem[Nisan and Szegedy(1994)]{NS94}
N.~Nisan and M.~Szegedy.
\newblock On the degree of {Boolean} functions as real polynomials.
\newblock \emph{Computational Complexity}, 4:\penalty0 301--313, 1994.
\newblock \doi{10.1007/BF01263419}.

\bibitem[Nisan and Wigderson(1995)]{NW94}
N.~Nisan and A.~Wigderson.
\newblock On rank vs. communication complexity.
\newblock \emph{Combinatorica}, 15\penalty0 (4):\penalty0 557--565, 1995.
\newblock \doi{10.1007/BF01192527}.

\bibitem[Radhakrishnan and Sanyal(2016)]{RS}
J.~Radhakrishnan and S.~Sanyal.
\newblock The zero-error randomized query complexity of the pointer function.
\newblock In \emph{{FSTTCS} 2016}, pages 16:1--16:13, 2016.
\newblock \doi{10.4230/LIPIcs.FSTTCS.2016.16}.

\bibitem[Reichardt(2011)]{Reichardt11}
B.~Reichardt.
\newblock Reflections for quantum query algorithms.
\newblock In \emph{{SODA}}, pages 560--569, 2011.
\newblock \doi{10.1137/1.9781611973082.44}.

\bibitem[Rubinstein(1995)]{Rub95}
D.~Rubinstein.
\newblock Sensitivity vs. block sensitivity of {Boolean} functions.
\newblock \emph{Combinatorica}, 15\penalty0 (2):\penalty0 297--299, 1995.
\newblock \doi{10.1007/BF01200762}.

\bibitem[Snir(1985)]{Snir85}
M.~Snir.
\newblock Lower bounds on probabilistic linear decision trees.
\newblock \emph{Theoretical Computer Science}, 38:\penalty0 69--82, 1985.
\newblock \doi{10.1016/0304-3975(85)90210-5}.

\bibitem[Sun(2007)]{Sun07}
X.~Sun.
\newblock Block sensitivity of weakly symmetric functions.
\newblock \emph{Theoretical Computer Science}, 384\penalty0 (1):\penalty0 87--91, 2007.
\newblock \doi{10.1016/j.tcs.2007.05.020}.

\bibitem[Sun(2011)]{Sun11}
X.~Sun.
\newblock An improved lower bound on the sensitivity complexity of graph properties.
\newblock \emph{Theoretical Computer Science}, 412\penalty0 (29):\penalty0 3524--3529, 2011.
\newblock \doi{10.1016/j.tcs.2011.02.042}.

\bibitem[Sun et~al.(2004)Sun, Yao, and Zhang]{SYZ04}
X.~Sun, A.~C. Yao, and S.~Zhang.
\newblock Graph properties and circular functions: How low can quantum query complexity go?
\newblock In \emph{{CCC}}, pages 286--293, 2004.
\newblock \doi{10.1109/CCC.2004.1313851}.

\bibitem[Tal(2013)]{Tal13}
A.~Tal.
\newblock Properties and applications of {Boolean} function composition.
\newblock In \emph{{ITCS}}, pages 441--454, 2013.
\newblock \doi{10.1145/2422436.2422485}.

\bibitem[Tardos(1989)]{Tardos89}
G.~Tardos.
\newblock Query complexity, or why is it difficult to seperate ${NP^A} \cap co{NP^{A}}$ from ${P}^{A}$ by random oracles {A}?
\newblock \emph{Combinatorica}, 9\penalty0 (4):\penalty0 385--392, 1989.
\newblock \doi{https://link.springer.com/article/10.1007/BF02125350}.

\bibitem[Tur{\'{a}}n(1984)]{Turan84}
G.~Tur{\'{a}}n.
\newblock The critical complexity of graph properties.
\newblock \emph{Information Processing Letters}, 18\penalty0 (3):\penalty0 151--153, 1984.
\newblock URL \url{https://doi.org/10.1016/0020-0190(84)90019-X}.

\end{thebibliography}
\label{sec:biblio}

\newpage
\appendix

\section{Known lower bounds for complexity measures for the class of transitive function}
The following table represents the individual known separations and the known example for different complexity measures for the class of transitive function:
\begin{table}[!h]
    \centering
    \begin{tabular}{|c||c|c|}
        \hline
$\textbf{Measure}$ & $\textbf{known lower bounds}$ & $\textbf{Known example}$ \\
\hline
{$\dqc$} &$\Omega(\sqrt{N})$ &$O(\sqrt{N})$  \\
 &\cite{SYZ04} &\cite{SYZ04}\\
\hline
{$\roc$ } & $\Omega(\sqrt{N})$& $O(\sqrt{N})$ \\
 &  \cite{SYZ04}& \cite{SYZ04}\\
\hline
{$\rqc$} &$\Omega({N^{\frac{1}{3}}})$ & $O(\sqrt{N})$ \\
& $\mathsf{bs(f)}= O(\rqc(f))$& \cite{SYZ04}\\
\hline
$\mathsf{C}$ & $\Omega(\sqrt{N})$ & $O(\sqrt{N})$  \\
& & $Tribe(\sqrt{N},\sqrt{N})$ \\
\hline
$\mathsf{RC}$& $\Omega({N^{\frac{1}{3}}})$&$O(\sqrt{N})$   \\
& $\mathsf{bs(f)}= O(\mathsf{RC}(f))$
& $Tribe(\sqrt{N},\sqrt{N})$ \\
\hline
$\mathsf{bs}$ &$\Omega({N^{\frac{1}{3}}})$& $\Tilde{O}(N^{\frac{3}{7}})$ \\
 &\cite{Sun07} & \cite{Sun07}, \cite{Drucker11}\\
\hline
$\mathsf{s}$ &$\Omega({N^{\frac{1}{8}}})$ &  $\Theta(N^{\frac{1}{3}})$ \\
 &$\deg(f)= O(\mathsf{s(f)})^2$ & \cite{Chakraborty11}\\
\hline
{$\lambda$} &$\Omega({N^{\frac{1}{12}}})$ &   $\Theta(N^{\frac{1}{3}})$  \\
 &$\mathsf{C(f)}= O(\mathsf{\lambda(f)})^6$  &\cite{Chakraborty11} \\
\hline
{$\qec$ } &$\Omega(N^{\frac{1}{4}})$ & $O(\sqrt{N})$  \\
& $\qqc(f) = O(\qec(f))$ & \cite{SYZ04}\\
\hline
{$\deg$} &$\Omega({N^{\frac{1}{4}}})$ &  $O(\sqrt{N})$ \\
 &$\deg(f) = \Omega( \adeg(f))$ & \cite{SYZ04} \\
\hline
{$\qqc$} &$\Omega(N^{\frac{1}{4}})$ & $\Tilde{O}(N^{\frac{1}{4}})$\\
&\cite{SYZ04}& \cite{SYZ04} \\
\hline
$\adeg$ &$\Omega({N^{\frac{1}{4}}})$  &$\Tilde{O}(N^{\frac{1}{4}})$ \\
&\cite{KT16} & \cite{SYZ04} \\
\hline
    \end{tabular}
    \caption{In each row, for the measure $A$, the two entries $a,b$ represents: 
(1) (Known lower bound) for all transitive Boolean function $f$,
$A(f) = \Omega(a)$, and  (2) (Known example) there exists a \emph{transitive} function $g$ such that $A(g) = O(b)$, where $a$ and $b$ are some polynomial in $N$.}
    \label{tab:my_label 3}
\end{table}

\end{document}